\documentclass[one column]{IEEEtran}
\usepackage{dcolumn}

\usepackage{epsfig}
\usepackage[cmex10]{amsmath}
\usepackage{amssymb}
\usepackage{amsthm}
\usepackage{amsfonts}
\usepackage{hyperref}
\usepackage{subcaption}
\usepackage{setspace}
\usepackage{bm}

\usepackage{tikz}

\newcolumntype{R}{D {.}{.}{1,0}}

\usepackage{algpseudocode}
\usepackage{algorithm}

\input{Definition.def}

\newcommand{\figref}[1]{Fig.\,\ref{#1}}

\ifCLASSOPTIONcompsoc
  \usepackage[nocompress]{cite}
\else
  \usepackage{cite}
\fi
\ifCLASSINFOpdf
\else
\fi
\begin{document}
%
\title{Berrut Approximated Coded Computing: Straggler Resistance Beyond Polynomial Computing}
%
%
%
%

\author{Tayyebeh~Jahani-Nezhad, Mohammad~Ali~Maddah-Ali
\IEEEcompsocitemizethanks{\IEEEcompsocthanksitem Tayyebeh Jahani-Nezhad is with  the  Department of Electrical Engineering, Sharif University of Technology, Tehran, Iran.\protect\\
E-mail: tayyebeh.jahaninezhad@ee.sharif.edu
\IEEEcompsocthanksitem Mohammad Ali Maddah-Ali is with  the  Department of Electrical Engineering, Sharif University of Technology, Tehran, Iran.\protect\\
Email:  maddah\_ali@sharif.edu}
\thanks{}}

\IEEEtitleabstractindextext{%
\begin{abstract}
One of the major challenges in using distributed learning to train complicated models with large data sets is to deal with stragglers effect.  As a solution, coded computation has been recently proposed to efficiently add redundancy to the computation tasks. In this technique, coding is used across data sets, and computation is done over coded data, such that the results of an arbitrary subset of worker nodes with a certain size are enough to recover the final results.   
The major drawbacks with those approaches are (1) they are limited to polynomial functions, (2) the number of servers that we need to wait for grows with the degree of the model, (3) they are not numerically stable for computation over real numbers.  In this paper, we propose Berrut Approximated Coded Computing (BACC), as an alternative approach, as a numerically stable solution, which works beyond polynomial functions computation and with any number of servers. 
The accuracy of the approximation is established theoretically and verified by simulation. In particular, BACC is used to train a deep neural network on a cluster of servers, which outperforms alternative uncoded solutions in terms of the rate of convergence.
\end{abstract}

}

\maketitle

\IEEEdisplaynontitleabstractindextext

%
\IEEEpeerreviewmaketitle

{\section{Introduction}}

%
%
%
%
Distributed machine learning is known as an inevitable solution to overcome the challenge of training complicated  models such as deep learning with large data sets \cite{dean2012large,4455,mapreduce}. In this solution, the data set or the model parameters are distributed among several servers and the tasks of training/evaluating are performed distributedly by those servers, in coordination with each other. In one scenario, for example, the parameters are maintained in a master node, while the data set is shared among some worker nodes. The worker nodes process data set locally and send the results to the master node to update the parameters. 
Distributed machine learning raises a list of challenges related to convergence rate, communication load, privacy, existence of faulty nodes, etc. One of the major challenges here is dealing with stragglers, or slow servers. Indeed, in those systems, the speed of computing is dominated by the speed of the slowest servers, as the master node needs to wait for all the worker nodes to complete their tasks \cite{tail}. One approach, proposed to deal with stragglers, is known as coded computing. In this approach, the computation is done over coded data, rather than raw data. Coding is used to efficiently add redundancy to the computing, such that from the results of a subset of servers, the final result can be calculated. This means that the master node does not have to wait for the results of stragglers to complete his task.
It is shown that coded computing can be effective in the context of distributed machine learning and can be applied in different problems such as coded distributed matrix multiplication and polynomial computations.
In \cite{high,speed}, one or both matrices are coded separately using MDS codes to compute matrix multiplication. In \cite{Polynomial}, \emph{polynomial codes}, and in \cite{entangle}, the extended version of these codes called \emph{entangled polynomial codes}, are proposed to code
each matrix with the desired partitioning, such that the number of unwanted computations is minimized. The general version of entangled polynomial
codes is proposed in \cite{opt-recovery} to multiply more than two matrices.  
In \cite{jahani2018codedsketch}, \emph{CodedSketch} as a straggler-resistant coded scheme is introduced to compute the approximation of matrix multiplication where the exact result of the multiplication is not required. \emph{Lagrange codes} \cite{lagrange} provide a novel strategy to compute an arbitrary polynomial function, without waiting for stragglers, and communication across worker nodes. Also, the privacy guarantees in secure multi-party computations are satisfied in this coded computing. As an application of Lagrange codes in distributed learning, a secure training process called \emph{CodedPrivateML} is proposed in \cite{codedpri}, which uses Lagrange codes to guarantee the privacy of data and the resulting logistic regression model. Since Lagrange coding technique is limited to computations with polynomial evaluation forms,
\emph{CodedPrivateML} uses the polynomial approximations of non-linearities of the model. Coding techniques can also be used to reduce the communication load in distributed learning problems in the presence of the stragglers \cite{Grad_Dimakis,ye2018communication,grad_app_dimakis,grad_app}.	 Existing coded computation approaches have some major challenges, which are a potential bottleneck in several important problems such as  distributed learning:
\begin{enumerate}
	\item
	They are limited to a specific class of functions such as matrix multiplications or polynomial functions, and a wide class of functions has not been considered.
	\item
	To calculate polynomial functions, the total number of servers needed is proportional to the degree of the polynomial times the size of the data set which can be prohibitively large.
	\item
	They are often designed for computations over finite fields. For computations over real numbers, those approaches face the serious problems in terms of computation instability. The reason is that they rely on Reed-Solomon decoding/Lagrange interpolation which are not suitable for real numbers.
\end{enumerate}
The main reason for the third challenge is that the decoding methods of those coded computation approaches are based on solving a system of linear equations with a Vandermonde matrix which is a matrix of coefficients. Real $n\times n$ Vandermonde matrices are ill-conditioned. More specifically, their condition number grows exponentially with $n$ \cite{gautschi1987lower}. There are some studies to overcome this issue. For example in \cite{fahim2019numerically}, instead of monomial basis, Chebyshev polynomials are used to develop a numerically stable approach for  polynomially coded computing such as matrix multiplication and Lagrange coded computing. Also, the condition number of the resulting decoding matrix is calculated in \cite{fahim2019numerically}, which is bounded polynomially in the number of nodes provided that the number of stragglers is a constant. In \cite{das2019random},  a convolutional coding approach is proposed to compute distributed matrix multiplications. Also,
a computable upper bound on the worst-case condition number is provided over all resulting decoding matrices. In \cite{ramamoorthy2019numerically},  the structured matrices such as circulant permutation matrices and rotation matrices are used as evaluation points of polynomials in coded distributed matrix computation problems. If the number of stragglers is a constant, then the worst-case condition numbers in this scheme grow polynomially in the number of nodes. Recently in \cite{soleymani2020analog}, analog Lagrange coded computing (ALCC) is proposed as an extension of Lagrange coded computing for the analog domain, where the operations in ALCC are done over the complex plane to evaluate polynomial functions. ALCC still faces the first two major challenges.  On the other hand, computing over real-valued data sets in the complex domain have greater required computations compared to the real domain.

In this paper, we propose  \emph{Berrut Approximated Coded Computing (BACC)} to mitigate major challenges of existing coded computation approaches. BACC  is a numerically stable coded scheme to approximately compute arbitrary functions, not necessarily polynomials, in a distributed setting consisting of a master node and some worker nodes. In this approach, the outcomes of any arbitrary subsets of available worker nodes are sufficient to calculate the approximated result, of course the more outcomes are received, the more accurate the final result will be. The error of this approximation is theoretically proven to be bounded. In addition, BACC is numerically stable with low computational complexity. 
This is also verified by simulation results. In particular, BACC is used to train a deep learning model, in which each worker node computes the gradients of desired functions based on specific combinations of some mini-batches of the data set. Having received computed gradients from a subset of available worker nodes, the master node can approximately decode the gradients in a numerically stable manner. The next iteration will be run after updating the parameters of the network.
Implementation results show that the proposed scheme outperforms repetitive computations in terms of the rate of convergence. 

The rest of this paper is organized as follows: In Section \ref{II} some preliminaries are reviewed. The proposed scheme  
is introduced in section \ref{propose}, and we detail the analytical guarantees in Section \ref{V}.
Section \ref{VI} describes our simulation results and Section \ref{LearningSec} represents the application of the proposed scheme in deep learning. Section \ref{VIII} shows the experiment results.
\subsection{Notation}
In this paper matrices and vectors are denoted by upper boldface letters and lower boldface letters respectively. $\mathcal{C}[a,b]$ denotes the space of all continuous functions on $[a,b]$, where $[a,b]$ is a closed interval. For $n_1, n_2\in\mathbb{Z}$ the notation $[n_1:n_2]$ represents the set $\{n_1,\dots n_2\}$. Also, $[n]$ denotes the set $\{0,\dots,n\}$ for $n\in\mathbb{Z}$. 
Furthermore, the cardinality of a set $\mathcal{S}$ is denoted by $|\mathcal{S}|$.  $\norm{f}$ denotes the maximum norm of function $f(x)$ over $x$ domain, i.e., $\norm{f}=\max_{x\in[a,b]}|f(x)|$. The $i$th element of a vector $\mathbf{v}$, is denoted by $[\mathbf{v}]_i$  and  the $(i,j)$-th entry of a matrix $\mathbf{A}$ is denoted by $[\mathbf{A}]_{i,j}$.
\section{Preliminaries}\label{II}
In this subsection, we review some preliminaries which are needed in the following sections.
\begin{definition}[Lagrange Polynomial Interpolation \cite{lagrange}] \label{defLag} Consider a set of $n+1$ distinct interpolation points $\mathcal{X}_n=\{x_i\}_{i=0}^n$, where $a\le x_0< x_1<\dots<x_n\le b$, for some real numbers $a,b$. In  addition, assume that $f_i=f(x_i), i\in[n]$, as the samples of a function $f\in \mathcal{C}[a,b]$ are given. \emph{Lagrange polynomial interpolation}  finds a polynomial $p_{\text{Lag}}\in\Pi_n$ which interpolates $f$ at $x_i$  i.e., $p_{\text{Lag}}(x_i)=f_i$, where $\Pi_n$ is the set of all polynomials of degree at most $n$ with real coefficients. In this method, $p_{\text{Lag}}(x)$ can be uniquely written as 
	\begin{align}\label{lag_int}
	p_{\text{Lag}}(x)=\sum_{i=0}^{n}{f_i\ell_{i,\text{Lag}}(x)}=L(x)\sum_{i=0}^{n}{\frac{w_if_i}{x-x_i}},
	\end{align} 
	where  $L(x)\triangleq\prod_{k=0}^{n}{(x-x_k)}$, and $\ell_{j,\text{Lag}}(x)$ for $j\in[n]$ are the Lagrange basis functions defined as 
	\begin{align}
	\ell_{j,\text{Lag}}(x)\triangleq\frac{\prod_{k=0,k\ne j}^{n}{(x-x_k)}}{\prod_{k=0,k\ne j}^{n}{(x_j-x_k)}}, j\in[n].	
	\end{align}
	The weights $w_i$ corresponding to $x_i$ are calculated as
	\begin{align}
	w_i=\frac{1}{\prod_{k=0,k\ne j}^{n}{(x_i-x_k)}}, i\in[n].\label{w}
	\end{align}	
\end{definition}
\begin{remark}
	According to Definition \ref{defLag}, Lagrange polynomial interpolation needs $\mathcal{O}(n^2)$ floating point operations to evaluate $p_{\text{Lag}}(x)$ at some $x\in\mathbb{R}$.	
\end{remark}
The representation of Lagrange interpolation in \eqref{lag_int} can be modified in such a way that it needs $\mathcal{O}(n)$ floating point operations to be evaluated. Assume the Lagrange interpolation of a constant function $g(x)=1$.  So, according to \eqref{lag_int} we have 
\begin{align}\label{1}
1=\sum_{i=0}^{n}\ell_{i,\text{Lag}}(x)=L(x)\sum_{i=0}^{n}{\frac{w_i}{x-x_i}}.
\end{align}
The new representation for Lagrange polynomial interpolation is obtained after dividing \eqref{lag_int} by \eqref{1} and canceling the factor $L(x)$ which is a common factor in the numerator and denominator.
\begin{definition}[Barycentric Polynomial Interpolation \cite{bary2004}] Another representation of Lagrange interpolation formula is called \textit{Barycentric polynomial interpolation} which is expressed as  
	\begin{align}\label{bary}
	p_{\text{Bary}}(x)=\sum_{i=0}^{n}\frac{{\frac{w_i}{(x-x_i)}}}{\sum_{j=0}^{n}\frac{w_j}{(x-x_j)}}f_i,
	\end{align}     
	where $w_i,i\in[n]$ is still defined by \eqref{w}.    
\end{definition}
\begin{remark}\label{remark2}
	Since the weights $w_i$ appear in both numerator and denominator of \eqref{bary}, any constant common factors in the weights can be factored out and canceled out without affecting the value of $p_{\text{Bary}}(x)$. It is one of the advantages of barycentric formula which avoids overflows and underflows in the weights computation.
\end{remark}
\begin{remark}
	Only $\mathcal{O}(n)$ operations are required for each evaluation of $p_{\text{Bary}}(x)$ in barycentric interpolation formula. Also, \eqref{bary} allows us to include additional interpolation points more easily and is more stable than the Lagrangian interpolation formula for a given point set \cite{bary2004}. 
\end{remark}
\begin{remark}
	There exist several explicit formulas for $w_i,i\in[n]$ for some particular sets of nodes. For example, assume the set of interpolation points are chosen from  \textit{Chebyshev points of the first kind} as 
	\begin{align}
	x_j=\cos{\frac{(2j+1)\pi}{2n+2}}, j=[n].
	\end{align}
	In \cite{henric}, it is shown that after eliminating the constant factors independent of $j$, the weights are computed as 
	\begin{align}
	w_j=(-1)^j\sin{\frac{(2j+1)\pi}{2n+2}}, j=[n].
	\end{align} 
	Another choice for the interpolation points is the \textit{Chebyshev points of the second kind} as 
	\begin{align}
	x_j=\cos{\frac{j\pi}{n}},j=[n].
	\end{align}
	In \cite{Salzer}, it is shown 
	\begin{align}\label{wcheby}
	w_j={(-1)}^j\delta_j, \hspace{1cm} \delta_j=\begin{cases}
	& \frac{1}{2},\hspace{2mm}\text{ if } j=0 \text{ or}\hspace{1mm} j=n, \\ 
	& 1, \hspace{2mm} \text{ otherwise}; 
	\end{cases}
	\end{align}
	Also, if the equidistant nodes on the interval $[-1,1]$ are chosen as the interpolation points then the weights will be calculated by $w_j={(-1)}^j\binom{n}{j}$ which is exponential by $n$ \cite{henric}. Thus, equidistant nodes are improper nodes for large values of $n$. That means polynomial interpolation in equidistant nodes is ill-conditioned.
\end{remark}
\begin{remark}
	If those interpolation points are chosen in the interval $[a,b]$ instead of $[-1,1]$, the original formula for the weights are only multiplied  by the constant factor $2^{n}{(b-a)}^{-n}$ which can be dropped in barycentric formula according to Remark \ref{remark2}. 
\end{remark}
\begin{remark}\label{remark6}
It is well known in the approximation theory that any other interpolation points, clustered at the endpoints of the interval $[-1,1]$ and distributed  asymptotically with density  $\frac{1}{\sqrt{1-x^2}}$ (as $n\to\infty$), can be proper. Thus, the calculated weights $w_i$ in \eqref{w} corresponds to these interpolation points do not grow exponentially by $n$ \cite{bary2004}. 
\end{remark}
It has been known that rational functions can also be used for interpolation to overcome some of the deficiencies of polynomial interpolation. 
 Since rational interpolation uses
	rational functions, and thus includes polynomial interpolation as a special case,  thus the resulting approximations would outperform polynomial interpolations. \cite{berrut2005recent,pachon2012fast}.
\begin{definition}[Rational Interpolant \cite{werner1986}]
	Let the set of $n+1$ distinct interpolation points $\mathcal{X}_n=\{x_i\}_{i=0}^{n}$
	be given with samples of a real value function $f\in \mathcal{C}[a,b]$ at these nodes, i.e., $f_i=f(x_i), i\in[n]$. The basic rational interpolant is defined as 
	\begin{align}\label{eq9}
	r(x)=\frac{p_{\tilde{m}}(x)}{q_{\tilde{n}}(x)}\in \mathcal{R}_{\tilde{m},\tilde{n}},
	\end{align}
	where $r(x_i)=f_i$ for $i\in[n]$. Also, $\mathcal{R}_{\tilde{m},\tilde{n}}$ is  the set of all rational functions with numerator and denominator degrees of at most $\tilde{m}$ and $\tilde{n}$ respectively.  Since there are $\tilde{m}+\tilde{n}+1$ unknown coefficients for $r(x)$, we must have $\tilde{m}+\tilde{n}+1=n+1$.
\end{definition}
In the following the barycentric rational interpolant is defined as a specific representation for the rational interpolant.
\begin{definition}[Barycentric Rational Interpolation \cite{bary2004,berrut1997,werner,berrut1997matrices}]\label{defBaryRa} The barycentric interpolant formula in \eqref{bary} can be applied with an arbitrary set of non-zero weights $u_i,i\in[n]$. The resulting interpolant is a rational interpolant called  \textit{Barycentric Rational Interpolation} and defined as 
	\begin{align}\label{rational}
	r_{\text{Bary}}(x)=\sum_{i=0}^{n}\frac{{\frac{u_i}{(x-x_i)}}}{\sum_{j=0}^{n}\frac{u_j}{(x-x_j)}}f_i\in\mathcal{R}_{n,n},
	\end{align}
	for all $u_i\ne0$ and bounded $f_i=f(x_i)$. Note that this interpolation has no restriction on calculating the weights $u_i$ by the distribution of points.  
\end{definition}
\begin{remark}
	Any rational interpolation of function $f$ using $f_i=f(x_i)$ for $i\in[n]$ can be expressed in barycentric rational form with some weights $u_i,i\in[n]$ \cite{berrut2005recent}.
\end{remark}

\begin{definition}[Berrut's Rational Interpolant \cite{berrut1988rational}]
	According to Definition \ref{defBaryRa}, the following rational function 
	\begin{align}\label{berrut}
	r_{\text{Berrut}}(x)=\sum_{i=0}^{n}\frac{{\frac{{(-1)}^i}{(x-x_i)}}}{\sum_{j=0}^{n}\frac{{(-1)}^j}{(x-x_j)}}f_i,
	\end{align}
	is called \textit{Berrut's Rational Interpolant} which interpolates $f_k$ at $x_k, k\in[n]$. The basis functions of this interpolant is denoted by
	\begin{align}
	\ell_{i,\text{Berrut}}=\frac{\frac{{(-1)}^i}{(x-x_i)}}{\sum_{j=0}^{n}\frac{{(-1)}^j}{(x-x_j)}}, i\in[n].
	\end{align} 
\end{definition}
\begin{lemma}\label{denum}
	Let the set of $n+1$ distinct interpolation points $\mathcal{X}_n=\{x_i\}_{i=0}^{n}$
	be given such that $x_0<x_1<\dots<x_n$, and $L(x)=\prod_{k=0}^{n}{(x-x_k)}$. Then the polynomial $q(x)=L(x)\sum_{j=0}^{n}{\frac{{(-1)^j}}{x-x_j}}$ has no real root \cite{berrut1988rational}.
\end{lemma}
\begin{remark}
	In rational interpolations, it is difficult to control the occurrence of  poles in the interval of interpolation.  According to Lemma \ref{denum}, $r_{\text{ Berrut}}(x)$ has no pole in the real line for any distribution of the interpolation points.
\end{remark} 
\begin{remark}
	An interpolation point $x_j$ is called unattainable if the interpolation condition  is not satisfied, i.e., $\frac{p(x_j)}{q(x_j)}\neq f(x_j)$ in \eqref{eq9}. Occurring unattainable points is one of the major flaws of traditional rational interpolants which is not occurred in Berrut's rational interpolant. 
\end{remark}
The error of an interpolation, and its numerically stability  are two important factors which are discussed later after reviewing some definitions such as  \textit{Lebesgue constant} which is one of the best criteria to determine which interpolation point sets are good.
\begin{definition}[Lebesgue Constant \cite{berrut1997,Cheney}]\label{Leb}
	Let $\mathcal{X}_n=\{x_j\}_{j=0}^n$ be a set of distinct interpolation points in the interval $[a,b]$ and  $\mathcal{B}_n=\{\ell_i\}_{i=0}^n$ be a set of basis functions of an interpolant. Assume $\mathcal{L}_{(\mathcal{X}_n,\mathcal{B}_n)}$ is a linear projection, which associates to any continuous function $f\in \mathcal{C}[a,b]$ the unique rational (polynomial) function, i.e., $\mathcal{L}_{(\mathcal{X}_n,\mathcal{B}_n)}f=\frac{p_n}{q_n}\in\mathcal{R}_{n,n}$. Thus, the Lebesgue constant is defined as 
	\begin{align}\label{lebDef}
	\Lambda_n\triangleq\norm{\mathcal{L}_{(\mathcal{X}_n,\mathcal{B}_n)}}=\sup_{f\in \mathcal{C}[a,b]}\frac{\norm{\mathcal{L}_{(\mathcal{X}_n,\mathcal{B}_n)}f}}{\norm{f}},
	\end{align}
	where $\norm{.}$ denotes the maximum norm. 
\end{definition}
\begin{lemma}[\cite{berrut1997}]
	Considering the basis functions $\mathcal{B}_n=\{\ell_i\}_{i=0}^n$, \eqref{lebDef} can be expressed as 
	\begin{align}\label{Lambda}
	\Lambda_n= \max_{x\in[a,b]}\sum_{i=0}^{n}{|\ell_i(x)|},
	\end{align}
	where $\Lambda_n(x)\triangleq\sum_{i=0}^{n}{|\ell_i(x)|}$ is called \emph{Lebesgue function}. 
\end{lemma}
According to the definition, the best choices of interpolation points are the ones that have a small Lebesgue constant for an interpolant  \cite{rivlin1981introduction}.

\begin{theorem}[Lebesgue Theorem \cite{berrut1997,shortCourse}]\label{LebesgueTh}
	Suppose there is a set of basis functions of an interpolant $\mathcal{B}_n=\{\ell_i\}_{i=0}^n$, and a set of interpolation points ${\mathcal{X}}_n=\{{x}_i\}_{i=0}^n$  in the interval $[a,b]$. Then,
	the error of the approximation of any function $f\in \mathcal{C}[a,b]$  using the  rational interpolant, i.e., $\mathcal{L}_{(\mathcal{X}_n,\mathcal{B}_n)}f=\frac{p_n}{q_n}\in\mathcal{R}_{n,n}$ is bounded from above as 
		\begin{align}\label{error}
	\norm{f-\mathcal{L}_{(\mathcal{X}_n,\mathcal{B}_n)}f}\le (1+\Lambda_n)\min\limits_{r\in\mathcal{Q}}\norm{f-r},
	\end{align}
	where $\mathcal{Q}$ is defined as the set of all rational functions in $\mathcal{R}_{n,n}$  passing through the interpolation points, with denominator $q_n$.
\end{theorem}
According to Theorem \ref{LebesgueTh}, for a particular interpolant, if we have properly distributed interpolation points that cause the smaller Lebesgue constant, then we will have a better interpolation of $f$. 

Now, consider Berrut's interpolation as a rational interpolant of $f_k$ at $x_k,k\in[n]$. There are several results based on the Lebesgue constant computation of Berrut's interpolant in different sets of interpolation points, which prove that Berrut's interpolation is extremely well-conditioned \cite{bos2013bounding,zhang2014improved}. In the following, we review an important result in this context. 
\begin{definition}[Well-Spaced Points \cite{bos2013bounding}]\label{wellSpaced}
	Let $\mathcal{X}_n=\{x_i\}_{i=0}^{n}$ be a set of ordered distinct interpolation points. Consider a family of sets, i.e., $\mathcal{X}=(\mathcal{X}_n)_{n\in\mathbb{N}}$, if there exist constants $C,R\ge1$ such that the following conditions 
	\begin{align*}
	(1)\hspace{0.15cm} &\frac{x_{k+1}-x_k}{x_{k+1}-x_j}\le\frac{C}{k+1-j},&&\text{for}j=[k],k=[n-1],\\
	(2)\hspace{0.15cm}&\frac{x_{k+1}-x_k}{x_j-x_k}\le\frac{C}{j-k}&&\text{for} j=[k+1:n],k=[n-1],\\
	(3)\hspace{0.15cm}&\frac{1}{R}\le\frac{x_{k+1}-x_k}{x_k-x_{k-1}}\le R,&&\text{for}k=[1:n-1],
	\end{align*}
	are satisfied, then $\mathcal{X}=(\mathcal{X}_n)_{n\in\mathbb{N}}$ is called a family of well-spaced points. Note that $R$ and $C$ must be independent of $n$.
\end{definition}
\begin{theorem}[\cite{bos2013bounding}]\label{lebesgue}
	Suppose we have a family of well-spaced points $\mathcal{X}=(\mathcal{X}_n)_{n\in\mathbb{N},n\ge2}$, with constant parameters $R,C\ge1$, where $\mathcal{X}_n=\{x_i\}_{i=0}^n$ is the set of interpolation points in the interval $[a,b]$. If Berrut's interpolant in \eqref{berrut} is used to interpolate function $f\in \mathcal{C}[a,b]$ at the points $\mathcal{X}_n$, then the Lebesgue constant under these assumption is bounded as 
	\begin{align}
	\Lambda_n\le(R+1)(1+2C\ln n).
	\end{align} 
\end{theorem}
Next we explain about the approximation error and convergence rate of the Berrut's rational interpolant.
\begin{theorem}[\cite{floater2007barycentric}]\label{error1}
	Assume Berrut's rational interpolant $r_{\text{Berrut}}$ in \eqref{berrut} as the interpolation formula for a continuous function $f\in \mathcal{C}[a,b]$ with second derivative (i.e., $f\in \mathcal{C}^2[a,b]$). Then, we have 
	\begin{align*}
	\norm{r_{\text{Berrut}}(x)-f(x)}\le h(1+\lambda)(b-a)\frac{\norm{f^{\prime\prime}(x)}}{2},
	\end{align*}
	if $n$ is odd, and 
	\begin{align*}
	\norm{r_{\text{Berrut}}(x)-f(x)}\le h(1+\lambda)\bigg((b-a)\frac{\norm{f^{\prime\prime}(x)}}{2}+\norm{f^{\prime}(x)} \bigg),
	\end{align*}
	if $n$ is even. In these inequalities,  $h\triangleq\max_{0\le i\le n-1}({x_{i+1}-x_{i}})$ for ordered interpolation points. In addition, $\lambda$ is defined as 
	\begin{align*}
	\lambda\triangleq \max_{1\le i \le n-2}\min\{\frac{x_{i+1}-x_i}{x_i-x_{i-1}},\frac{x_{i+1}-x_i}{x_{i+2}-x_{i+1}} \},
	\end{align*}  
	and is referred as local mesh ratio.
\end{theorem}
Thus $r_{\text{Berrut}}(x)$ converges to $f(x)$ at the rate of $\mathcal{O}(h)$  under the assumption that $f\in \mathcal{C}^2[a,b]$, and provided that the local mesh ratio $\lambda$ is bounded as $h\to 0$ which depends on the distribution of interpolation points.

The error caused by floating-point arithmetic is really significant in interpolation problems. Numerical stability of an algorithm is a measure to determine the sensitivity of its output caused by small changes in the input data. 
\begin{definition}
	Assume $\tilde{f}$ is an interpolant which interpolates function $f$. If for some small  \textit{backward error} $\delta_1>0$, and for any $x\in\mathbb{R}$, there exist some $|\delta_x|\le\delta_1$ such that $	\tilde{f}(x)=f(x+\delta_x),$
	then, this interpolant is called $\delta_1$-\textit{backward stable}. In addition, if for some small \textit{forward error} $\delta_2>0$, we have  $\frac{\norm{\tilde{f}-f}}{\norm{f}}=\delta_2$, then, this interpolation is called $\delta_2$-\textit{forward stable}.
\end{definition}
In other words, a backward stable interpolant gives the right value
of $f$ at the approximately right value of $x$, and a forward stable interpolant
provides an output that is close enough to the right result.
\begin{remark}
	In \cite{higham}, it is shown that \eqref{bary} is forward stable for any set of interpolating nodes as long as the Lebesgue constant is not too large. Also, in\cite{mascarenhas2014backward}, it is shown that \eqref{bary} and \eqref{rational} are backward stable as well  when the constant $\Lambda_n$ remains small. 
\end{remark}
\subsection{An overview of Lagrange Coded Computing}
Lagrange coded computing (LCC) \cite{lagrange} is designed to calculate an arbitrary polynomial function $p(X)$ for $K$ inputs $X_0, \ldots, X_{K-1}$, over a cluster of $N+1$ servers. It is based on the following three steps: 
\begin{enumerate}
	\item
	The master node forms polynomial $u(z)$, such that $u(\alpha_k)=X_k$, using Lagrange interpolation, for some distinct values of $\alpha_k$, $k \in [K-1]$. 
	\item
	The master node calculates $u(\beta_n)$, and sends to worker node $n$ to calculate $p(u(\beta_n))$, $n\in [N]$, for some distinct values of $\beta_n$, $n \in [N]$.
	\item
	The master node recovers $g(z)=p(u(z))$, upon receiving $\deg(g(z))+1=(K-1)\deg(p(X))+1$ answers from the workers nodes. Then it calculates $p(X_k)$ as $g(\alpha_k)= p(u(\alpha_k))$, $k\in [K-1]$.
\end{enumerate}
 The advantage of this approach is that having the results of any arbitrary subsets of the workers nodes of size $(K-1)\deg(p(X))+1$, the master node can calculate $p(X_0), \ldots,p(X_{K-1})$. Thus it can tolerate up to $N- (K-1)\deg(p(X))$ stragglers. Lagrange coded computing suffers from several problems,  when it is used for computing over real numbers. 
\begin{enumerate}
	\item Its application is limited to polynomial computations. 
	\item 	In Lagrange coded computing, the total number of worker nodes that the master node needs to wait for to recover the final result is proportional to the degree of the polynomial times the size of the input data set, which can be prohibitively large. In other words, if the number of non-straggling worker nodes is less than $(K-1)\deg{f}+1$, the final results can not be computed.
	\item It is originally designed for computations over the finite field. This approach is not proper for computation over real numbers and faces serious problems in terms of computation instability.
\end{enumerate}
\section{The Proposed Scheme}\label{propose}
As explained, existing coded computing approaches have some major  challenges in distributed computing. To overcome these challenges, we propose \emph{Berrut Approximated Coded Computing} to approximately  evaluate \emph{any arbitrary function} using a distributed system when the data and all operations are in the field of \emph{real numbers}. This scheme is  \emph{numerically stable}  with \emph{low computational complexity}, which can be used in problems such as distributed learning. In this scheme, we propose a different encoding and decoding method. We also suggest particular points for encoding the input data set and assigning the tasks to the worker nodes. The accuracy of the approximation established theoretically and verified by simulation results in different settings, such as distributed learning problems.
\\The objective is to approximately evaluate ``\emph{an arbitrary function}" $f:\mathbb{V}\to\mathbb{U}$ over an input data set $\mathbf{X}=(\mathbf{X}_0,...,\mathbf{X}_{K-1})$ in a numerically stable manner with bounded errors, where $\mathbb{V}$ and $\mathbb{U}$ are the set of the real matrices.  A distributed system  with one master node and $N+1$ worker nodes $\mathcal{W}_0,\dots\mathcal{W}_N$ is utilized to approximately compute the evaluation of $f$ over data set $\mathbf{X}$, i.e., $\tilde{\mathbf{Y}}_i\approx f(\mathbf{X}_i)$ for $i\in[K-1]$. Also, assume that in the distributed system there may be some straggling worker nodes. The proposed straggler resistant scheme is based on the following steps:\\
\textbf{Step 1.} 
The master node creates the coded data $\hat{\mathbf{X}}_i=\mathcal{E}_i(\mathbf{X})$ and assigns it to $i$th worker node,  where $\mathcal{E}_i$ is an encoding function that maps the raw data $(\mathbf{X}_0,...,\mathbf{X}_{K-1})$ to the coded matrix $\hat{\mathbf{X}}_i$ for $i$th worker node. More precisely, the master node forms the following rational function $u:\mathbb{R}\to\mathbb{V}$ 
\begin{align}\label{coding}
u(z)=\sum_{i=0}^{K-1}{ \frac{{\frac{{(-1)}^i}{(z-\alpha_i)}}}{\sum_{j=0}^{K-1}\frac{{(-1)}^j}{(z-\alpha_j)}}\mathbf{X}_i},
\end{align}
for some distinct values $\alpha_0,\dots\alpha_{K-1}\in \mathbb{R}$.
One can verify that $u(\alpha_j)=\mathbf{X}_j$, for all $j\in[K-1]$.\\
In this scheme, we suggest to choose  $\alpha_j, j\in[K-1]$,  as \emph{Chebyshev points of the first kind} as 
\begin{align}
\alpha_j=\cos(\frac{(2j+1)\pi}{2K}),\hspace{1cm}j\in[K-1]. 
\end{align}
\textbf{Step 2.} 
The master node assigns $\hat{\mathbf{X}}_i=u(z_i)$ to $i$th worker node to apply $f$ on $\hat{\mathbf{X}}_i$ and send the result back. In the proposed scheme, we suggest to choose $z_i, i\in[N]$, as \emph{Chebeshev points of the second kind}, i.e., 
\begin{align}
z_i=\cos{\frac{i\pi}{N}}, \hspace{1cm}i\in[N].
\end{align}	
Having received $\hat{\mathbf{X}}_i$ from the master node, the $i$th worker node computes $\hat{\mathbf{Y}}_i=f(\hat{\mathbf{X}}_i)$. Then it returns the result to the master node.\\
\textbf{Step 3.} 	
The master node waits for the results from the set of fastest worker nodes, denoted by $\mathcal{F}$. Then it approximately calculates  $f(\mathbf{X}_i)$, $i\in[K-1]$, from $\big\{\hat{\mathbf{Y}}_j\big\}_{j\in\mathcal{F}}$, using the decoding function $\mathcal{D}\big(\big\{\hat{\mathbf{Y}}_j\big\}_{j\in\mathcal{F}},\mathcal{F}\big)$. The decoding function is based on Berrut's rational interpolant, with computational complexity of $\mathcal{O}(|\mathcal{F}|)$.\\
In other words, ‌the master node, after receiving outcomes of non-straggling worker nodes, creates a rational function which approximately interpolates $f(u(z))$  as 
\begin{align}\label{inter}
r_{\text{Berrut},\mathcal{F}}(z)=\sum_{i=0}^{n}\frac{{\frac{{(-1)}^i}{(z-\tilde{z}_i)}}}{\sum_{j}\frac{{(-1)}^j}{(z-\tilde{z}_j)}}f(u(\tilde{z}_i)),
\end{align} 
where $\tilde{z}_i\in\mathcal{S}, i\in[n]$ are the interpolation points, where $\mathcal{S}=\{\cos{\frac{j\pi}{N}}, j\in \mathcal{F}\}$, and $n\triangleq |\mathcal{F}|-1$.  
Now the master node then approximately computes $f(\mathbf{X}_i)\approx r_{\text{ Berrut},\mathcal{F}}(\alpha_i),i\in[K-1]$.
\begin{remark}
	In this scheme, there is no strict notion of recovery threshold or the minimum number of required computation results from worker nodes. The master node uses the available results of the computations sent by non-straggling worker nodes and computes the final results. The more the number of results is, the more accurate the final result will be.
\end{remark}
\begin{remark}
	The application of BACC is not limited to polynimial functions, and this scheme can be used to approximately evaluate any arbitrary functions. 
\end{remark}
\begin{remark}
	In this scheme, we suggest to choose  $\alpha_j, j\in[K-1]$,  as Chebyshev points of the first kind, and we suggest to choose $z_i, i\in[N]$, as Chebyshev points of the second kind. 
	In Theorem \ref{myWellSpaced}, we will prove that the Lebesgue constant for Berrut's rational interpolant grows logarithmically in the size of  \emph{a subset of Chebyshev points}.
\end{remark}
\begin{remark}
	In BACC, we suggest using Berrut's rational interpolant rather than barycentric interpolant for the decoding step.  Because of the stragglers, the master node faces a subset of Chebyshev points as the interpolation points rather than the entire set. If we had the entire set, calculating $w_i$ would have a well-behaved explicit formula as \eqref{wcheby}. However, when we have a subset of them, we need to use the general formula \eqref{w} to calculate $w_i$. Using \eqref{w} itself is not numerically stable in practice. The reason is that, according to  Remark \ref{remark6}, any subset of Chebyshev points is not necessarily a set of properly distributed interpolation points for polynomial interpolants. Thus, we use Berrut's rational interpolant.
\end{remark}
\begin{remark}
	In the polynomial interpolation, the errors caused by floating-point arithmetic are significant, and the barycentric formula has a  good performance in this respect. However, the barycentric representation is not well-conditioned for some distribution of the interpolation points. In particular, even barycentric interpolation in equidistant points faces strange behavior called Runge phenomenon \cite{runge1901empirische}, which is a problem of large oscillations near the endpoints.
	In such cases, no matter what formulation is used, polynomial interpolation is not recommended for interpolation.  Thus, for Lagrange coded computing, in the encoding step, we recommend to use barycentric interpolation. In addition, the popular equidistant points are not recommended.
\end{remark}
\section{Analytical Guarantees}\label{V}
In order to guarantee that the proposed interpolation points and the approximation result are acceptable, we establish the following theorems.
\begin{lemma}\label{lemma6}
	Assume $\mathcal{X}_n=\{x_j\}_{j=0}^{n}$ is a subset of  $\tilde{\mathcal{X}}_{N}=\{\tilde{x}_k\}_{k=0}^{N}$ with $n+1$ elements such that $x_0>x_1>\dots>x_{n}$, where $n=N-s$, and $\tilde{x}_k,k\in[N]$ are the Chebyshev points of the second kind, i.e., $\tilde{x}_k=\cos{\frac{k\pi}{N}},k\in[N]$, and $s$ is a constant number independent of $N$. The Lebesgue function associated with Berrut’s interpolant in $\mathcal{X}_n=\{x_j\}_{j=0}^{n}$ attains its maximum if there exist $\bar{k}$ such that $x_{j}=\tilde{x}_{j}=\cos{\frac{j\pi}{N}}$ for $j\in [\bar{k}]$ and $x_{j}=\tilde{x}_{j+s+1}=\cos{\frac{(j+s+1)\pi}{N}}$ for $j\in[\bar{k}+1:N-s]$, i.e., that all $s$ elements not included in $\mathcal{X}_n$ are ordered consecutively in $\tilde{\mathcal{X}}_{N}$.
\end{lemma}
\begin{proof}
	Lemma \ref{lemma6} expresses that the worst case in the interpolation step of the proposed scheme is occurred when $s$ straggling worker nodes correspond to the consecutive elements of $\tilde{\mathcal{X}}_{N}=\{\tilde{x}_k\}_{k=0}^{N}$. The proof of Lemma \ref{lemma6} can be found in  supplementary materials.
\end{proof}
\begin{theorem}\label{myWellSpaced}
		Let $\mathcal{X}_n=\{x_j\}_{j=0}^{n}$ be a subset of  $\tilde{\mathcal{X}}_{N}=\{\tilde{x}_k\}_{k=0}^{N}$ with $n+1$ elements such that $x_0>x_1>\dots>x_{n}$, where $n=N-s$, $s$ is a constant number independent of $N$, and $\tilde{x}_k,k\in[N]$, are the Chebyshev points of the second kind. Then, $\mathcal{X}=(\mathcal{X}_n)_{n\in\mathbb{N}}$ is a family of well-spaced points with $C=\frac{\pi^2(s+1)}{2}$ and $R=\frac{(s+1)(s+3)\pi^2}{4}$ for $s<N-2$. In addition, the Lebesgue constant for Berrut's rational interpolant in $\mathcal{X}_n$ is upper bounded as 
\begin{align*}
	\Lambda_n\le\big(\frac{(s+1)(s+3)\pi^2}{4}+1\big)\hspace{-1mm}\big(1+\pi^2(s+1)\ln (N-s)\big).
	\end{align*}
\end{theorem}
\begin{proof}
	This proof is based on Definition \ref{wellSpaced} and Theorem \ref{lebesgue}, and shows that the Lebesgue constant for the proposed scheme is bounded above by $c\ln{(N-s)}$ for some constant $c>0$. The formal proof can be found in  supplementary materials.
%
\end{proof}
As an example, Lebesgue function for Berrut's rational interpolant in the proposed scheme with different values for parameters $N$ and $s$ is demonstrated in  \figref{fig1}.
	\begin{figure*}[!t]
	\centering
	\includegraphics[width=\columnwidth]{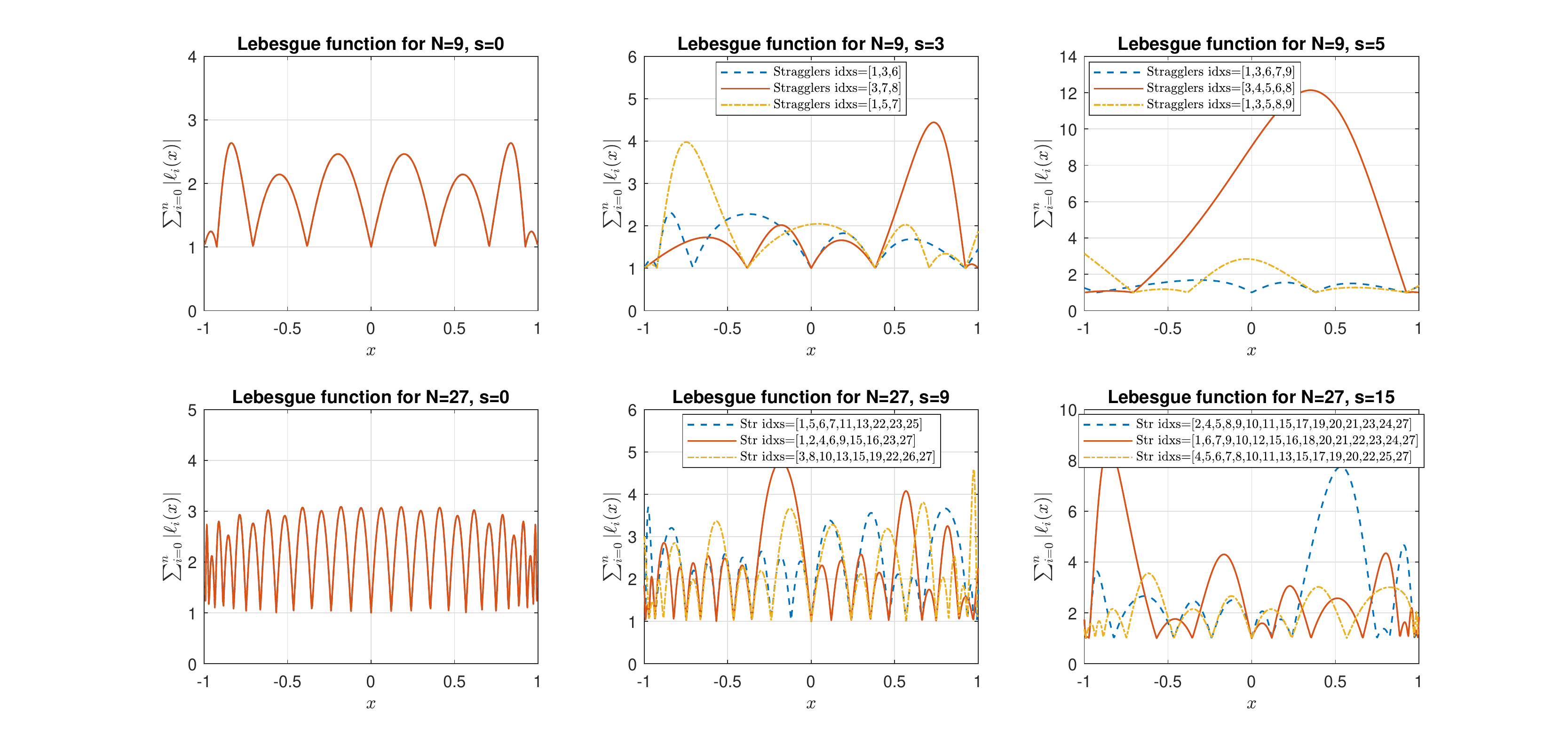}
	\caption{The value of Lebesgue function versus $x\in[-1,1]$ in Berrut's rational interpolant with different parameters value $N$ and $s$.}
	\label{fig1}
\end{figure*}
Note that the Lebesgue constant is not a function of the evaluation points or the function $f$. One way to bound the approximation error of the proposed method is to use the Lebesgue constant as follows.
\begin{corollary}
	Consider a distributed setting, consisting of $N+1$ worker nodes with up to $n+1=N+1-s$ non-straggling worker nodes with corresponding $z_j$, $j\in [n]$, interpolation points. Then, the error of the approximately evaluation of any arbitrary function $f$ using the proposed scheme is upper bounded as
	\begin{align}\label{eq32}
	\norm{r_{\text{Berrut},\mathcal{F}}(z)-g(z)}\le \big(1+\Lambda_n\big)\min\limits_{r(z)\in \mathcal{Q}}\norm{g(z)-r(z)},
	\end{align}
	where  $g(z)\triangleq f(u(z))$ and $u(z)$ is defined in \eqref{coding}.
	$\mathcal{Q}$ is defined as the set of all rational functions in the form of
	  $r(z)=p(z)/q(z)$, where $r(z_j)=g(z_j)$, $j\in[n]$, $p(z)$ is a polynomial function of degree $n$, and $q(z)=L(z)\sum_{j=0}^{n}{\frac{{(-1)^j}}{z-z_j}}$. Also, $\norm{.}$ denotes the maximum norm, i.e., $\norm{f}=\max_{x}|f(x)|$.
\end{corollary}
\begin{proof}
	Inequality \eqref{eq32} is derived using Theorem \ref{LebesgueTh}, Theorem \ref{lebesgue} and Theorem \ref{myWellSpaced}.
\end{proof}
Note that bound \eqref{eq32} is not tight, and further analysis is required to derive better upper bounds.\\
In the proposed approach, the outcomes of any arbitrary subsets of available worker nodes are sufficient to calculate the approximated result of $g(z)$ with bounded approximation error introduced in \eqref{eq32}. According to \eqref{eq32}, the more outcomes are received from worker nodes, the more accurate the approximation will be. 
\begin{theorem}\label{thBound}
	Let $r_{\text{Berrut},\mathcal{F}}(z)$  be defined by  \eqref{inter} and $g(z)=f(u(z))$ have a continuous
	second derivative on $[-1,1]$.  In a system with $N+1$ worker nodes and $s$ stragglers, where $s<N-2$, the approximation error of this interpolation using BACC is upper bounded as 
	\begin{align*}
	\norm{r_{\text{Berrut},\mathcal{F}}(z)-g(z)}\le2(1+R)\sin\big({\frac{(s+1)\pi}{2N}}\big)\norm{g^{\prime\prime}(z)},
	\end{align*} 
	if $N-s$ is odd, and  
	\begin{align*}
	\norm{r_{\text{Berrut},\mathcal{F}}(z)-g(z)}\hspace{-1mm}\le\hspace{-1mm}2(1\hspace{-1mm}+\hspace{-1mm}R)\sin(\hspace{-1mm}{\frac{(s\hspace{-1mm}+\hspace{-1mm}1)\pi}{2N}}\hspace{-1mm})\hspace{-1mm}\bigg(\hspace{-1mm}\norm{g^{\prime\prime}(z)}\hspace{-1mm}+\hspace{-1mm}\norm{g^{\prime}(z)}\hspace{-1mm}\bigg),
	\end{align*}
	if $N-s$ is even, where $R=\frac{(s+1)(s+3)\pi^2}{4}$. 
\end{theorem}
\begin{proof}
	The proof can be found in  supplementary materials.
\end{proof}
\begin{remark}
	Theorem \ref{thBound} shows that in the proposed scheme, for a fixed total number of worker nodes $N$, the fewer stragglers exist, the more accurate the final result will be. 
\end{remark}

\section{Simulation Results}\label{VI}
As mentioned before, the proposed scheme can be used to approximately evaluate arbitrary real functions at the desired data points. In this section, we demonstrate the performance of the proposed scheme through some simulation results.\\
\textbf{Case 1:} 
$f:\mathbb{R}^{m\times m}\to\mathbb{R}^{m\times m}$ is a polynomial function of degree $\deg f$ and the goal is to evaluate $f$ over a data set $\mathbf{X}=(\mathbf{X}_0,\dots,\mathbf{X}_{K-1})$, where $\mathbf{X}_i\in\mathbb{R}^{m\times m }$ for $i\in [K]$. Note that in this case $N$ is not necessarily greater than $(K-1)\deg f +1$. Indeed, $N$ can be very smaller than $(K-1)\deg f +1$.

Recall that if we use Lagrange coded computing, in this case, we need at least $(K-1)\deg f +1$ worker nodes; otherwise, the scheme does not work. Even if the number of worker nodes is greater than $(K-1)\deg f +1$, LCC is not numerically stable. The reason is that Lagrange coded computing relays on Vandermonde matrices for decoding, and a real-valued $n\times n$ Vandermonde matrix is ill-conditioned, specially when $n$ becomes large. On the other hand, in Lagrange coded computing, if the evaluation of a high degree polynomial function over a small data set is considered, the number of required servers becomes prohibitively large.
In many applications, having an approximated results of evaluation $f$ over the data set is enough as long as it is numerically stable, and the computational complexity is low. Thus,  proposed BACC scheme can be used to approximately evaluate the function over the desired data set without those challenges.
In this simulation, for each value of  $(N,k,\deg{f})$, we generate 100 different polynomial functions $f:\mathbb{R}\to\mathbb{R}$, where the coefficients of each polynomial function are chosen uniformly at random on the interval $[-10,10]$. The input data points $X_1\dots X_{K-1}\in\mathbb{R}$ are chosen uniformly at random on the interval $[-1,1]$ for each function. We consider scenarios where $s$ of $N$ worker nodes are stragglers. Since there are many different subsets of size $s$ from $N$, we generate 1000 cases chosen uniformly at random for each polynomial function. In particular, in Fig. \ref{fig3} we consider four cases: 
$(N,K,\deg{f})=(500,20,25),(700,20,35),(700,30,35),(500,30,25)$. For each number of stragglers, the shadow area in Fig. \ref{fig3} represents
the relative error of all functions averaged over all choices of stragglers. Also, the solid line represents the overall average of these errors.
Figure \ref{fig3} shows that the number of required non-straggling worker nodes is not necessarily equal or greater than $(K-1)\deg{f}+1$ to approximately evaluate function $f$ over the input data set. As mentioned before, Lagrange coded computing does not work with less than $(K-1)\deg{f}+1$ worker nodes. Also, the more results are available from the worker nodes, the more accurate is the final results.

%
%
\begin{figure}
	\begin{center}
		\centerline{\includegraphics[width=0.5\columnwidth]{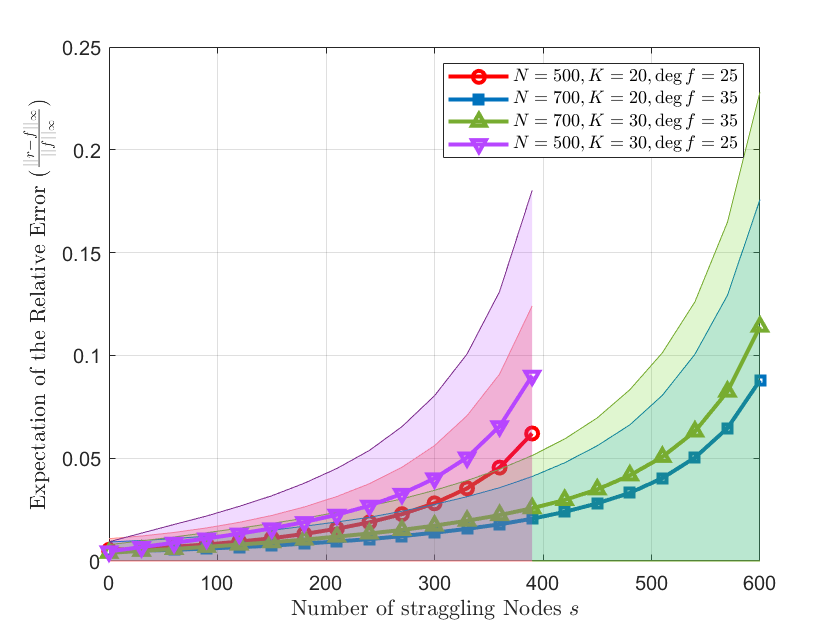}}
		\caption{Expectation of the relative error of the 100 polynomial functions $f$ of degree $\deg{f}$ using the proposed scheme. Note that in many cases, the number of non-straggling worker nodes is less than $(K-1)\deg{f}+1$, and still the error is reasonable.}
		\label{fig3}
	\end{center}
\end{figure}
Another choice for the interpolation points is the equidistant points. Figure \ref{fig7} compares the impact of using BACC with Chebyshev points and using BACC with the equidistant points in the expectation of relative error results over a set of 100 different polynomial functions. In this simulation, we consider $(N,K,\deg{f})=(500,20,25)$. Figure \ref{fig7} shows that, as compared to equidistant points, using BACC with Chebyshev points can reduce the expectation of the relative error by an order of magnitude, where $f$ is a polynomial function.
\begin{figure}
	\begin{center}
		\centerline{\includegraphics[width=0.5\columnwidth]{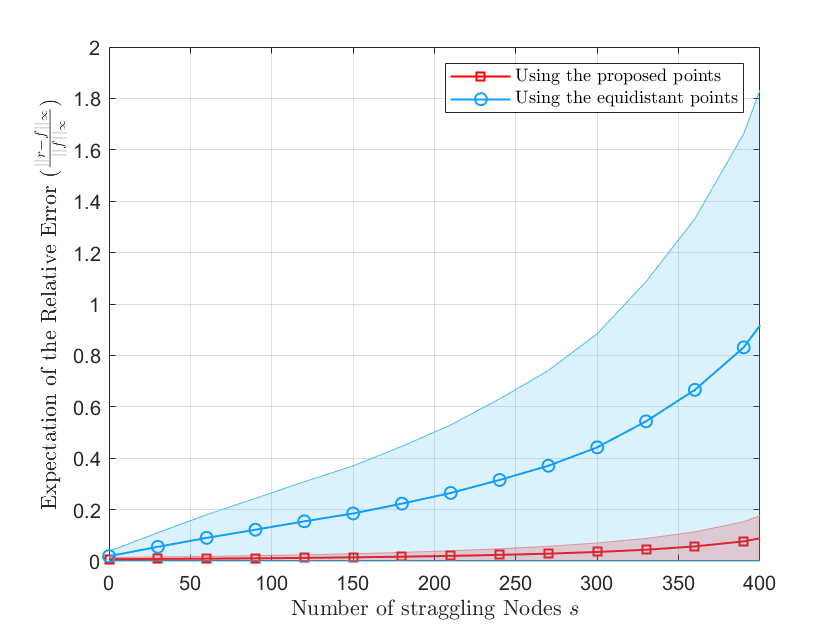}}
		\caption{Comparison between the impact of using the proposed scheme and using the equidistant points in the expectation of the relative error of a set of 100 polynomial functions $f$ of degree $\deg{f}=25$, where $N=500$ and $K=20$.}
		\label{fig7}
	\end{center}
\end{figure}
\\\textbf{Case 2:} $f:\mathbb{R}^{m\times m}\to\mathbb{R}^{m\times m}$ is not a polynomial function. Note that the existing coded computing schemes are limited to polynomial functions, and  they do not work in this case. Here, we use BACC to approximately evaluate non-polynomial function $f$ over the input data set $\mathbf{X}_0,\dots\mathbf{X}_{K-1}$. We consider $f=x\sin{x}$ and the input points are $X_i=-12+\frac{24i}{19}$, for $i=0,\dots,19$. The performance of the proposed scheme is shown in Fig. \ref{fig11}, where $N=60$ and $s=20$.

Figure \ref{fig12} shows the expectation of the relative error of the approximation of the function $f=x\sin{x}$ using the proposed scheme versus the number of stragglers. In this figure, two different values for the total number of worker nodes $N=250,300$ and the number of input data points $K=20,30$ are considered. Note that the stragglers are chosen uniformly at random over $N$ worker nodes in 1000 iterations. Also, the input data points are chosen uniformly at random in the interval $[-1,1]$. As shown in Fig. \ref{fig12}, the proposed scheme exhibits a very good performance for this function. 
\begin{figure}
	\begin{center}
		\centerline{\includegraphics[width=0.5\columnwidth]{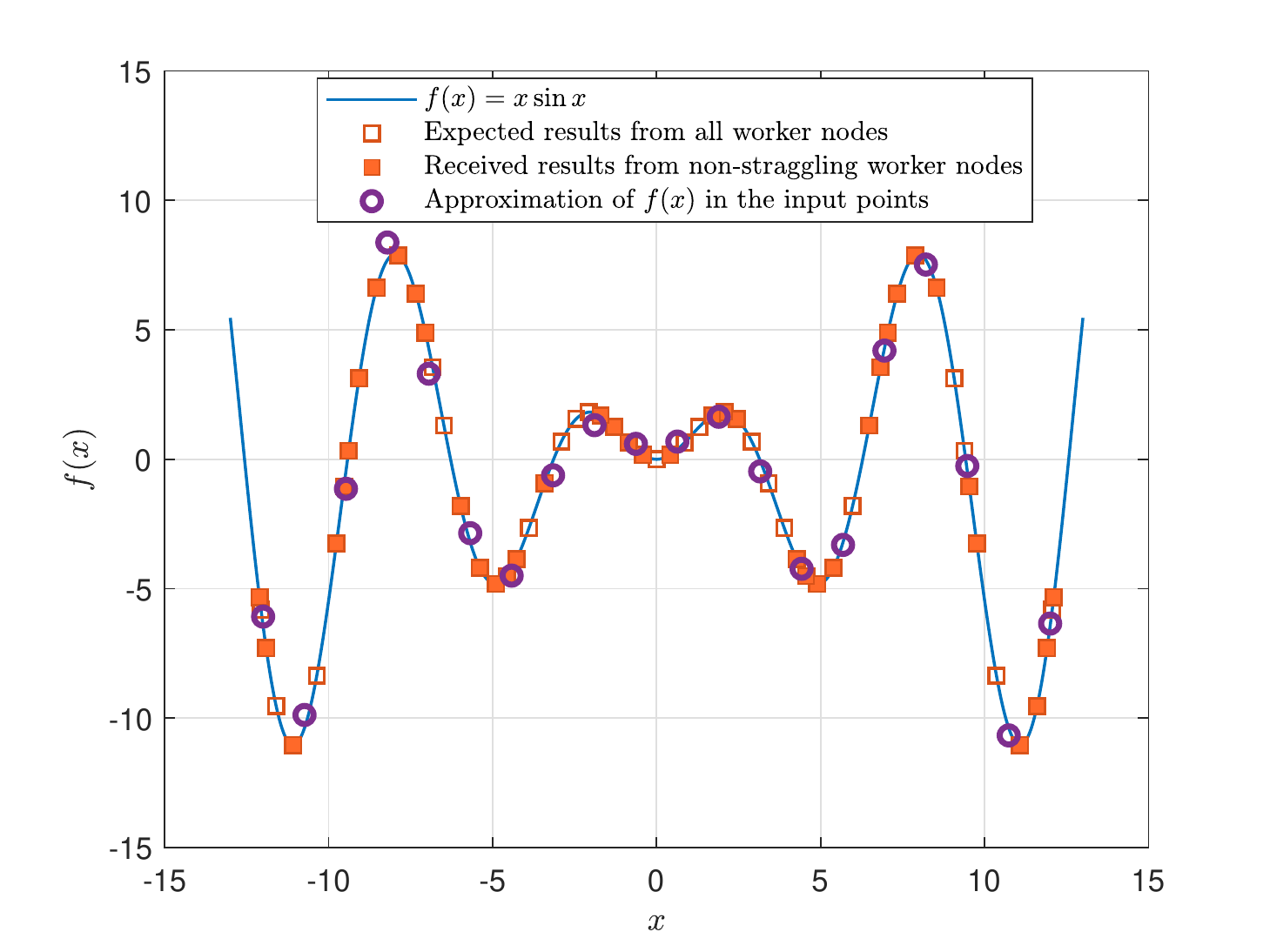}}
		\caption{Approximation of the function $f=x\sin{x}$ in the input data set using BACC, where $K=20$, $N=60$ and $s=20$.
		}
		\label{fig11}
	\end{center}
\end{figure}
\begin{figure}[htb]
	\centering
		\includegraphics[width=0.5\columnwidth]{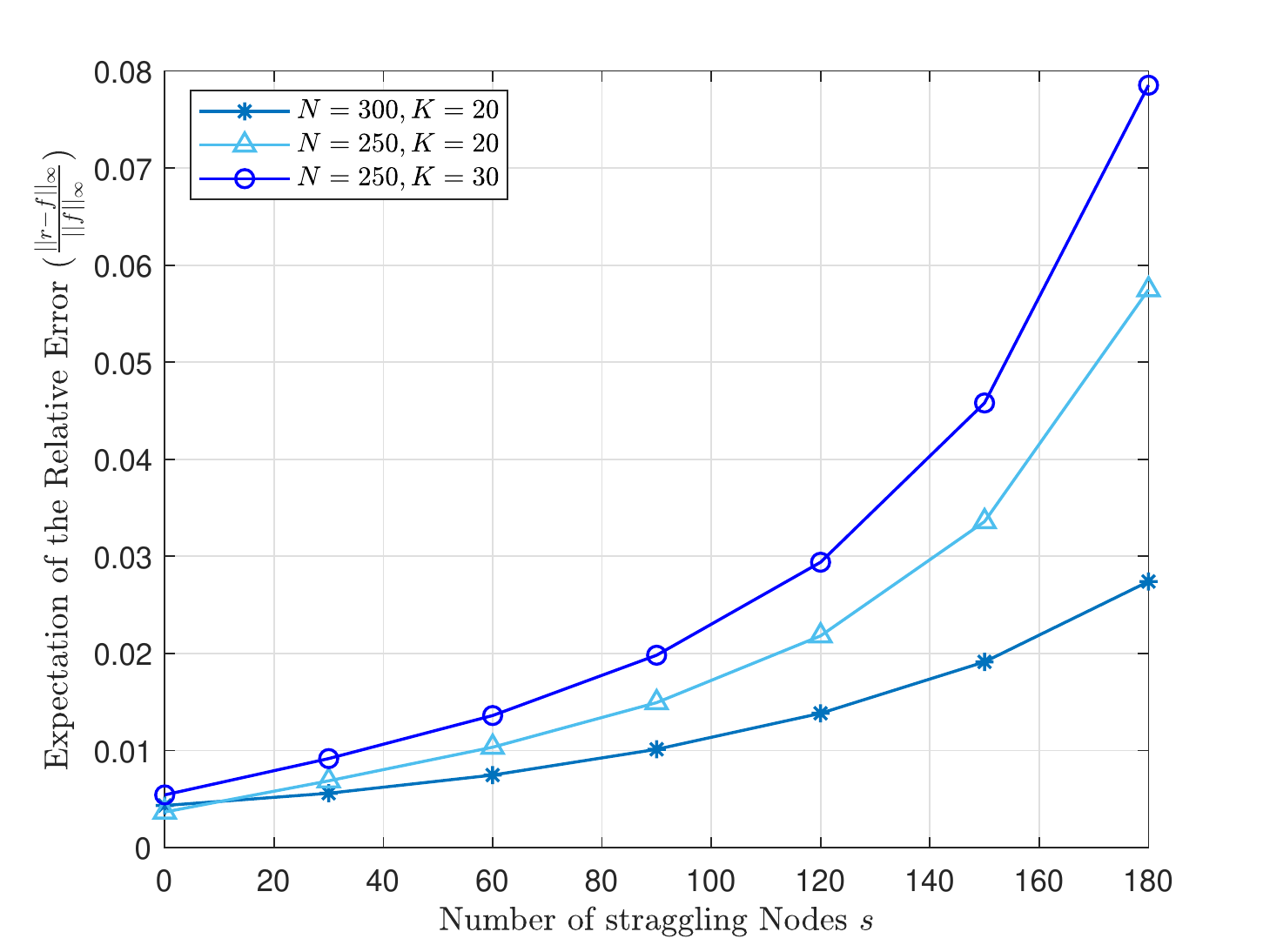}
	\caption{Expectation of the relative error of the approximation of the function $f=x\sin{x}$ for $x\in[-1,1]$ using the proposed scheme, for different values of $N$ and $K$.}
	\label{fig12}
\end{figure}

\section{Application of the Proposed Scheme in Deep Learning}\label{LearningSec}
Deep neural networks face the challenge of training complicated models with large data sets. Distributed machine learning can be used as an inevitable solution to overcome this challenge. In this scenario, the data set is divided among worker nodes, and the stochastic gradient descent algorithm is used to train the model. In each iteration of the training process, partial gradients are computed in each worker node based on the local data samples and are returned to the master node, where the model parameters are updated using these gradients. The updated parameters are then reported to the master node. In this section, BACC as a coding scheme, is used to overcome some challenges of distributed learning such as stragglers effect, and also some major challenges of coded computation approaches in distributed learning such as numerical instability, limiting to a specific class of function like polynomial functions, and the increase of the number of needed servers in proportion to the degree of the polynomial and the size of data set. 

In our scheme, worker nodes compute the partial gradient on their assigned coded data set. Having aggregated the results of fastest worker nodes, the master node is capable of approximately evaluating the full gradient even if there are $s$ stragglers in the distributed system. In brief, BACC is a numerically stable scheme in which the full gradient vectors are approximately computed with low computational complexity. 
Also, in BACC, the number of required worker nodes is decreased by sacrificing a small amount of accuracy such that the approximated result has a bounded error.

Assume a deep neural network (DNN) with $L$ layers consisting of $L$ parameter matrices $\mathbf{W}_\ell$ (weights) for $\ell=1,\dots L$, which have to be updated in each iteration during the training process. Each iteration of training process has three steps called feedforward, back-propagation, and updating step. Now we explain briefly about these steps which are needed to explain our scheme. 
Suppose the $\ell$th layer of DNN has $M_\ell$ neurons, and let $\mathbf{w}_{\ell}(m)\in\mathbb{R}^{M_{\ell-1}}$ be the weight vector of $m$th neuron of the $\ell$th layer where $m=1,\dots M_\ell$ and $\ell=1,\dots L$. Thus, the output of the  $m$th neuron of the $\ell$th layer at the $t$th iteration is computed as
\begin{align}
{s}_{\ell}^{(t)}(m)=f_{a}\bigg((\mathbf{w}_{\ell}^{(t)}(m))^T\mathbf{s}_{\ell-1}^{(t)}\bigg)\in\mathbb{R},
\end{align}
where $f_a(.)$ is an activation function which can be Sigmoid, ReLU, $\tanh$ or other non-linear common activation functions. Note that for the first layer we have $\mathbf{s}_0^{(t)}=\mathbf{x}^{(t)}$, where $\mathbf{x}^{(t)}\in\mathbb{R}^d$ is the training data sample with $d$ features used for the $t$th iteration of training.
Now consider a supervised machine learning problem. Given a training data set $\mathcal{D}=\{(\mathbf{x}_i,y_i)\}_{i=1}^n$, where $\mathbf{x}_i\in\mathbb{R}^d$ is the input sample with $d$ features, and $y_i\in\mathbb{R}$ is the corresponding label. We represent the input samples as a matrix $\mathbf{X}\in\mathbb{R}^{n\times d}$, where $\mathbf{x}^T_i$ is the $i$th row of $\mathbf{X}$. In many supervised machine learning problems, the goal is to learn the parameters $\mathbf{W}_\ell\in\mathbb{R}^{M_\ell\times M_{\ell-1}}$ for $\ell=1\dots L$ by minimizing the following empirical  loss function
\begin{align*}
J(\mathcal{D};\mathbf{W}_1,\dots \mathbf{W}_L)=\frac{1}{|\mathcal{D}|}\sum_{(\mathbf{x},y)\in \mathcal{D}}{J(\mathbf{x},y;\mathbf{W}_1,\dots \mathbf{W}_L)}.
\end{align*}
Common loss functions $J(\mathbf{x},y;\mathbf{W}_1,\dots \mathbf{W}_L)$ in machine learning problems are mean squared error loss, hinge loss, logistic loss and cross-entropy loss,  which are chosen according to some factors such as the type of machine learning algorithm, the type of data set and complexity of this minimization problem.
One approach to solve this optimization problem is gradient descent algorithm, which starts with some initial value for $\mathbf{W}_\ell$, and then in each iteration $t$ updates this parameters as 
\begin{align*}\label{update}
\mathbf{w}_\ell^{(t+1)}(m)=\mathbf{w}_\ell^{(t)}(m)-\eta\nabla_{\mathbf{w}_\ell^{(t)}(m)}{J(\mathcal{D};\mathbf{W}_1^{(t)},\dots \mathbf{W}_L^{(t)})},
\end{align*}
where $\eta\in\mathbb{R}$ is the learning rate, $\mathbf{w}_{\ell}^{(t)}(m)\in\mathbb{R}^{M_{\ell-1}}$ is the weight of $m$th neuron of the $\ell$th layer at the $t$th iteration, and $\nabla_{\mathbf{w}_\ell^{(t)}(m)}{J(\mathcal{D};\mathbf{W}_1^{(t)},\dots \mathbf{W}_L^{(t)})}$ is the gradient of the loss function at the current parameters. The gradient is computed as 
\begin{align*}
\nabla_{\mathbf{w}_\ell^{(t)}(m)}{J(\mathcal{D};\mathbf{W}_1^{(t)},\dots \mathbf{W}_L^{(t)})}=\frac{1}{\mathcal{|D|}}\sum_{i=1}^{n}\frac{\partial J(\mathbf{x}_i,y_i;\mathbf{W}_{1:L}^{(t)})}{\partial {\mathbf{w}_\ell^{(t)}(m)} },
\end{align*}
where the partial gradient operates on each scalar element of the vector ${\mathbf{w}_\ell^{(t)}(m)}$. According to partial gradients of the loss function, we have
\begin{align}
\mathbf{w}_\ell^{(t+1)}(m)=\mathbf{w}_\ell^{(t)}(m)-\frac{\eta}{|\mathcal{D}|}\sum_{i=1}^{n}{\delta_\ell^{(t,i)}(m)\mathbf{s}_{\ell-1}^{(t,i)}},
\end{align}
where $\delta_\ell^{(t,i)}(m)$ is the back-propagation error of neuron $m$ of $\ell$th layer corresponding to the data sample $i$ which can be computed as a function of back-propagation error vector of $\ell+1$th layer.
 \begin{figure*}
	
	\begin{center}
		\scalebox{0.6}{%

			\tikzset{every picture/.style={line width=0.75pt}} 
			
			\begin{tikzpicture}[x=0.75pt,y=0.75pt,yscale=-1,xscale=1]
			
			\draw (545.7,106.7) node  {\includegraphics[width=51.45pt,height=42.45pt]{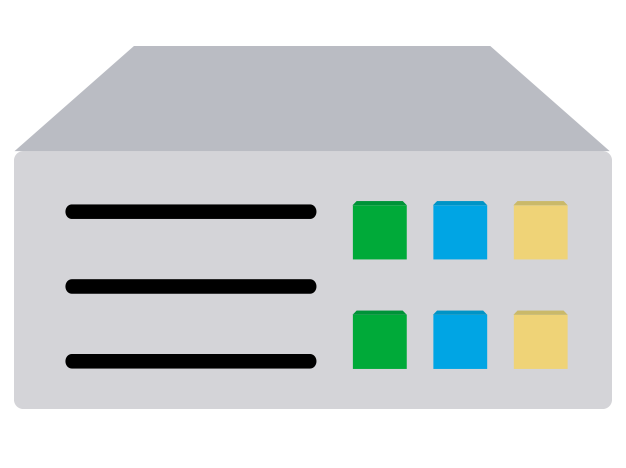}};
			\draw (434.95,222.8) node  {\includegraphics[width=48.53pt,height=37.5pt]{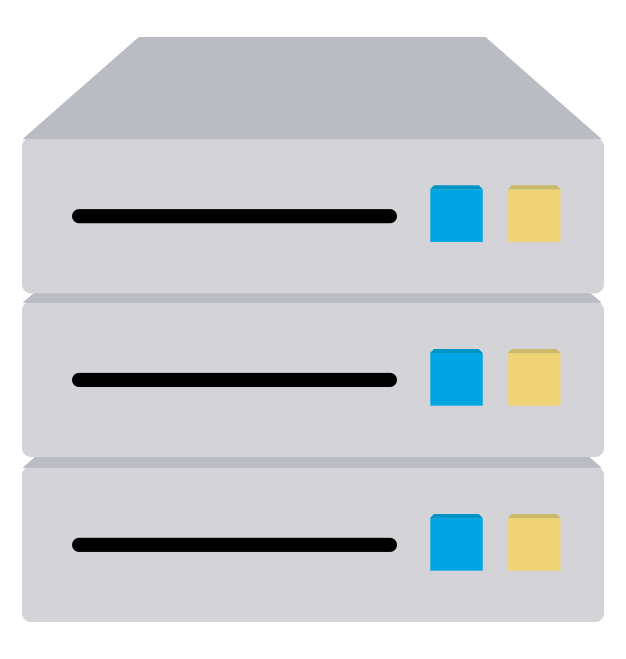}};
			\draw (547.35,223.2) node  {\includegraphics[width=48.53pt,height=37.5pt]{server3.png}};
			\draw (659.35,223.2) node  {\includegraphics[width=48.53pt,height=37.5pt]{server3.png}};
			\draw    (432.88,111.48) -- (432.88,199.98) ;
			\draw    (432.88,111.48) -- (511,110.97) ;
			\draw    (544.38,128.48) -- (545.88,200.48) ;
			\draw    (660.38,111.48) -- (660.38,199.98) ;
			\draw    (580.38,112.48) -- (660.38,111.48) ;
			Flowchart: Summing Junction [id:dp032187534651328065] 
			\draw  [color={rgb, 255:red, 208; green, 2; blue, 27 }  ,draw opacity=1 ][line width=4.5]  (641,225.19) .. controls (641,214.86) and (649.25,206.48) .. (659.44,206.48) .. controls (669.62,206.48) and (677.88,214.86) .. (677.88,225.19) .. controls (677.88,235.52) and (669.62,243.9) .. (659.44,243.9) .. controls (649.25,243.9) and (641,235.52) .. (641,225.19) -- cycle ; \draw  [color={rgb, 255:red, 208; green, 2; blue, 27 }  ,draw opacity=1 ][line width=4.5]  (646.4,211.96) -- (672.47,238.42) ; \draw  [color={rgb, 255:red, 208; green, 2; blue, 27 }  ,draw opacity=1 ][line width=4.5]  (672.47,211.96) -- (646.4,238.42) ;
			Image [id:dp7616296876686197] 
			\draw (553.89,262.1) node  {\includegraphics[width=14.69pt,height=19.5pt]{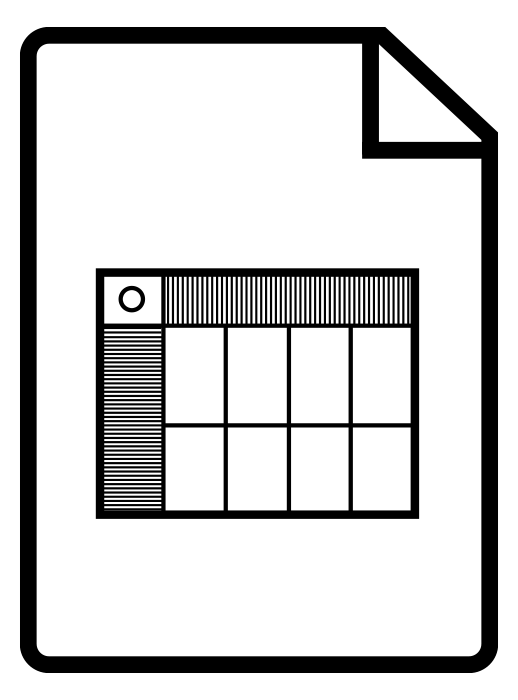}};
			\draw (663.09,263.3) node  {\includegraphics[width=14.69pt,height=19.5pt]{data1.png}};
			\draw (438.69,263.5) node  {\includegraphics[width=14.69pt,height=19.5pt]{data1.png}};
			\draw (926.42,105.1) node  {\includegraphics[width=51.45pt,height=42.45pt]{master.png}};
			\draw (815.67,221.2) node  {\includegraphics[width=48.53pt,height=37.5pt]{server3.png}};
			\draw (928.07,221.6) node  {\includegraphics[width=48.53pt,height=37.5pt]{server3.png}};
			\draw (1040.07,221.6) node  {\includegraphics[width=48.53pt,height=37.5pt]{server3.png}};
			\draw    (813.6,109.88) -- (813.6,198.38) ;
			\draw    (813.6,109.88) -- (891.72,109.37) ;
			\draw    (925.1,126.88) -- (926.6,198.88) ;
			\draw    (1041.1,109.88) -- (1041.1,198.38) ;
			\draw    (961.1,110.88) -- (1041.1,109.88) ;
			Flowchart: Summing Junction [id:dp1871779273865517] 
			\draw  [color={rgb, 255:red, 208; green, 2; blue, 27 }  ,draw opacity=1 ][line width=4.5]  (1021.72,223.59) .. controls (1021.72,213.26) and (1029.98,204.88) .. (1040.16,204.88) .. controls (1050.34,204.88) and (1058.6,213.26) .. (1058.6,223.59) .. controls (1058.6,233.92) and (1050.34,242.3) .. (1040.16,242.3) .. controls (1029.98,242.3) and (1021.72,233.92) .. (1021.72,223.59) -- cycle ; \draw  [color={rgb, 255:red, 208; green, 2; blue, 27 }  ,draw opacity=1 ][line width=4.5]  (1027.12,210.36) -- (1053.2,236.82) ; \draw  [color={rgb, 255:red, 208; green, 2; blue, 27 }  ,draw opacity=1 ][line width=4.5]  (1053.2,210.36) -- (1027.12,236.82) ;
			Image [id:dp1664369086004025] 
			\draw (934.61,260.5) node  {\includegraphics[width=14.69pt,height=19.5pt]{data1.png}};
			\draw (1043.81,261.7) node  {\includegraphics[width=14.69pt,height=19.5pt]{data1.png}};
			\draw (818.41,259.9) node  {\includegraphics[width=14.69pt,height=19.5pt]{data1.png}};
			\draw (819.21,287.9) node  {\includegraphics[width=14.69pt,height=19.5pt]{data1.png}};
			\draw (936.21,287.7) node  {\includegraphics[width=14.69pt,height=19.5pt]{data1.png}};
			\draw (1043.81,287.7) node  {\includegraphics[width=14.69pt,height=19.5pt]{data1.png}};
			\draw (180.7,106.18) node  {\includegraphics[width=51.45pt,height=42.45pt]{master.png}};
			\draw (70.35,224.68) node  {\includegraphics[width=48.53pt,height=37.5pt]{server3.png}};
			\draw (182.35,222.68) node  {\includegraphics[width=48.53pt,height=37.5pt]{server3.png}};
			\draw (294.35,222.68) node  {\includegraphics[width=48.53pt,height=37.5pt]{server3.png}};
			\draw    (67.88,110.96) -- (67.88,199.46) ;
			\draw    (67.88,110.96) -- (146,110.45) ;
			\draw    (179.38,127.96) -- (180.88,199.96) ;
			\draw    (295.38,110.96) -- (295.38,199.46) ;
			\draw    (215.38,111.96) -- (295.38,110.96) ;
			Flowchart: Summing Junction [id:dp2766789746196616] 
			\draw  [color={rgb, 255:red, 208; green, 2; blue, 27 }  ,draw opacity=1 ][line width=4.5]  (276,224.67) .. controls (276,214.34) and (284.25,205.96) .. (294.44,205.96) .. controls (304.62,205.96) and (312.88,214.34) .. (312.88,224.67) .. controls (312.88,235.01) and (304.62,243.38) .. (294.44,243.38) .. controls (284.25,243.38) and (276,235.01) .. (276,224.67) -- cycle ; \draw  [color={rgb, 255:red, 208; green, 2; blue, 27 }  ,draw opacity=1 ][line width=4.5]  (281.4,211.44) -- (307.47,237.9) ; \draw  [color={rgb, 255:red, 208; green, 2; blue, 27 }  ,draw opacity=1 ][line width=4.5]  (307.47,211.44) -- (281.4,237.9) ;
			Image [id:dp6934178136868703] 
			\draw (71.29,263.18) node  {\includegraphics[width=14.69pt,height=19.5pt]{data1.png}};
			\draw (189.69,263.58) node  {\includegraphics[width=14.69pt,height=19.5pt]{data1.png}};
			\draw (299.69,264.38) node  {\includegraphics[width=14.69pt,height=19.5pt]{data1.png}};
			
			\draw (414.4,260.52) node [anchor=north west][inner sep=0.75pt]  [font=\footnotesize]  {$x_{1}$};
			\draw (527.4,260.92) node [anchor=north west][inner sep=0.75pt]  [font=\footnotesize]  {$x_{2}$};
			\draw (636.2,260.72) node [anchor=north west][inner sep=0.75pt]  [font=\footnotesize]  {$x_{3}$};
			\draw (510.4,66.6) node [anchor=north west][inner sep=0.75pt]  [font=\footnotesize]  {$Master\ Node$};
			\draw (792.72,260.92) node [anchor=north west][inner sep=0.75pt]  [font=\footnotesize]  {$x_{1}$};
			\draw (908.92,260.32) node [anchor=north west][inner sep=0.75pt]  [font=\footnotesize]  {$x_{2}$};
			\draw (1018.12,260.12) node [anchor=north west][inner sep=0.75pt]  [font=\footnotesize]  {$x_{3}$};
			\draw (891.12,65) node [anchor=north west][inner sep=0.75pt]  [font=\footnotesize]  {$Master\ Node$};
			\draw (793.32,285.32) node [anchor=north west][inner sep=0.75pt]  [font=\footnotesize]  {$x_{2}$};
			\draw (908.92,285.52) node [anchor=north west][inner sep=0.75pt]  [font=\footnotesize]  {$x_{3}$};
			\draw (1017.52,285.92) node [anchor=north west][inner sep=0.75pt]  [font=\footnotesize]  {$x_{1}$};
			\draw (544,306.87) node [anchor=north west][inner sep=0.75pt]   [align=left] {(b)};
			\draw (928,306.07) node [anchor=north west][inner sep=0.75pt]   [align=left] {(c)};
			\draw (44,260.41) node [anchor=north west][inner sep=0.75pt]  [font=\footnotesize]  {$\tilde{x}_{1}$};
			\draw (161.6,260.21) node [anchor=north west][inner sep=0.75pt]  [font=\footnotesize]  {$\tilde{x}_{2}$};
			\draw (271.2,260.21) node [anchor=north west][inner sep=0.75pt]  [font=\footnotesize]  {$\tilde{x}_{3}$};
			\draw (145.4,66.08) node [anchor=north west][inner sep=0.75pt]  [font=\footnotesize]  {$Master\ Node$};
			\draw (174.6,305.75) node [anchor=north west][inner sep=0.75pt]   [align=left] {(a)};

			\end{tikzpicture}

		}
	\end{center}
	\caption{A simple comparison among three distributed learning settings  (a) Berrut Coded Computing, (b) the scheme without data redundancy, and (c) the data replication scheme, all with $N=3$, $K=3$, where each $\tilde{x}_i$ is a particular linear combination of $x_j,j\in[1:3]$ as mentioned in the proposed scheme. The number of straggling worker nodes in all settings is 1. }
	\label{fig9}
\end{figure*}
Now in a distributed learning approach consider a distributed system with one master node and $N+1$ worker nodes $\mathcal{W}_0,\dots\mathcal{W}_N$ which aim to collaboratively compute the gradient assuming that $s$ nodes are straggler. Due to the limited computing power of each worker node, the training dataset $\mathcal{D}$ is partitioned into $K$ non-overlapping equal-size subsets $\mathcal{D}=\{\mathcal{D}_0,\dots\mathcal{D}_{K-1}\}$, where $\mathcal{D}_j=(\mathbf{X}_j,\mathbf{y}_j)$ for $j\in[K-1]$  is a subset of dataset with size $|\mathcal{D}_j|=B$, for some integer $B$. Thus, the update rule for the weights of layer $\ell$ is given by
\begin{align}\label{update2}
\mathbf{W}_\ell^{(t+1)}=\mathbf{W}_\ell^{(t)}-\frac{\eta}{|\mathcal{D}|}\sum_{j=0}^{K-1}\Delta_\ell^{(t,j)}(\mathbf{S}_{\ell-1}^{(t,j)})^T,
\end{align}
where  $\Delta_\ell^{(t,j)}\in\mathbb{R}^{M_\ell\times B}$ is the back-propagation error matrix corresponding to the $j$th subset of dataset whose $(m,b)$-th entry is defined as $\delta_\ell^{(t,b_j)}(m)$ ($b_j$ denotes the index of data sample in the subset $j$ of dataset), and $\mathbf{S}_{\ell-1}^{(t,j)}\in\mathbb{R}^{M_{\ell-1}\times B}$ is the output of  layer $\ell-1$ corresponding to the $j$th subset of dataset whose $(m,b)$-th entry is ${s}_{\ell-1}^{(t,b_j)}(m)$.  Assume function $g(\mathbf{X},\ell;\mathbf{W}_1^{(t)},\dots \mathbf{W}_L^{(t)})\triangleq\sum_{j=0}^{K-1}\Delta_\ell^{(t,j)}(\mathbf{S}_{\ell-1}^{(t,j)})^T$  for given values $\ell$ and $\{\mathbf{W}_\ell^{(t)}\}_{\ell=1}^L$. It is clear that function $g$ is a non-linear function of $\mathbf{X}$. For simplicity of presentation, we denote $g(\mathbf{X},\ell;\mathbf{W}_1^{(t)},\dots \mathbf{W}_L^{(t)})$ by $g(\mathbf{X})$.

According to \eqref{update2}, we can apply the proposed scheme in 
Section \ref{propose} to approximately compute the value of the updated weights. In the following, we briefly describe our method in the distributed learning setting. \\
\textbf{1)} First, the master node encodes the subsets of the data samples, i.e., $\mathbf{X}_i,i\in[K-1]$ using \eqref{coding}, and generates the rational function $u(z)$. 
	Then the master node selects $z_r=\cos{\frac{j\pi}{N}}$ and sends $\hat{\mathbf{X}}_r\triangleq u(z_r)$ to the $r$th worker node for $r\in[N]$. 
	At each iteration $t$ of the training, the master node needs to send the current estimated parameters $\{\mathbf{W}^{(t)}_\ell\}_{\ell=1}^L$ to each worker node.\\
\textbf{2)} Each worker node stores a linear combination of all subsets of data set. 
	So, worker nodes compute the gradient based on the shared  parameter matrix $\{\mathbf{W}^{(t)}_\ell\}_{\ell=1}^L$ with their local data samples $\hat{\mathbf{X}}_r,r\in[N]$ in parallel, and then send the results back to the master node.\\ 
\textbf{3)}	The announced result of worker node $r$ is an evaluation of the function $g(u(z))$ at $z=z_r$.   Having received the results from a set of non-straggling worker nodes $\mathcal{F}$, the master node can approximately recover $g(u(z))$ with $\mathcal{O}(|\mathcal{F}|)$ of computational complexity as 
	\begin{align}\label{appGrad}
	\hat{g}(u(z))=\sum_{i=0}^{n}\frac{{\frac{{(-1)}^i}{(z-\tilde{z}_i)}}}{\sum_{j}\frac{{(-1)}^j}{(z-\tilde{z}_j)}}g(u(\tilde{z}_i)),
	\end{align}
	where $\tilde{z}_i\in\mathcal{S}, i\in[n]$ are the interpolation points, where $\mathcal{S}=\{\cos{\frac{j\pi}{N}}, j\in \mathcal{F}\}$, and $n\triangleq |\mathcal{F}|-1$. \\
\textbf{4)}	The approximated value of $g(\mathbf{X}_j)$  is achieved by computing $\hat{g}(u(\alpha_j))$  for $j\in[K-1]$, and the master node can update the model parameter using \eqref{update2}.

\begin{figure*}[htb]
	\centering
	\begin{tabular}{@{}c@{}}
		\includegraphics[trim=0 280 0 50,clip,width=\columnwidth]{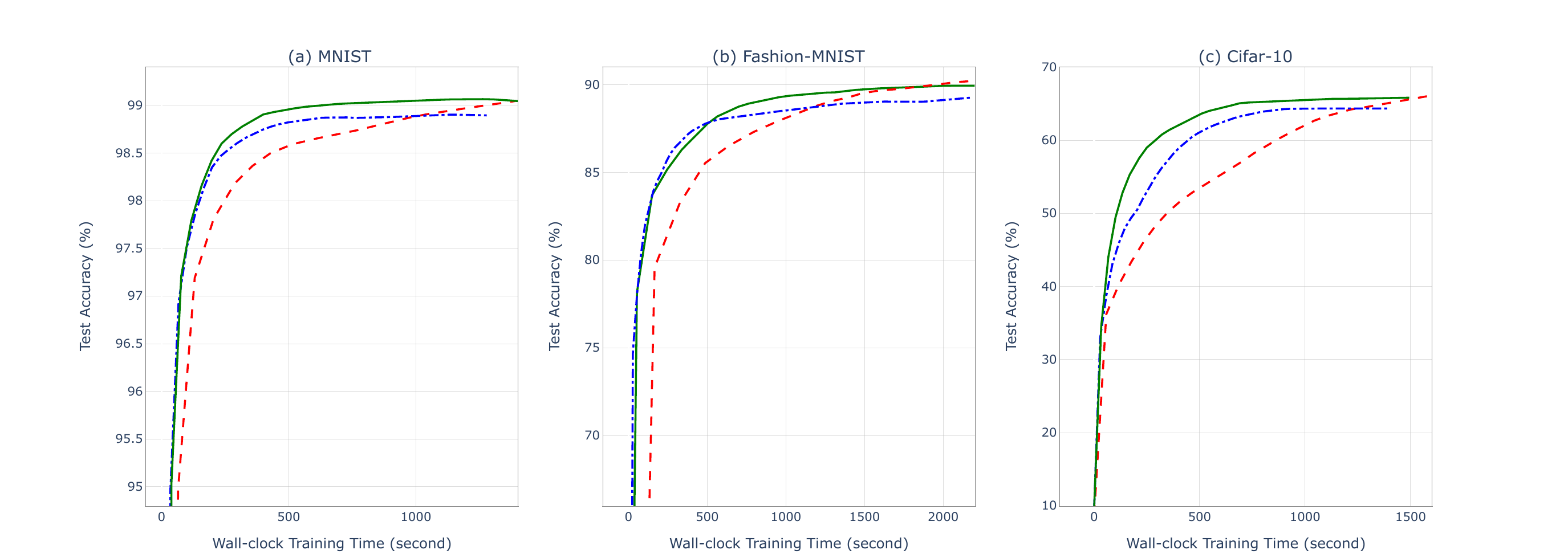}\\ 
		(I)\\
		\includegraphics[trim=0 250 0 50,clip,width=\columnwidth]{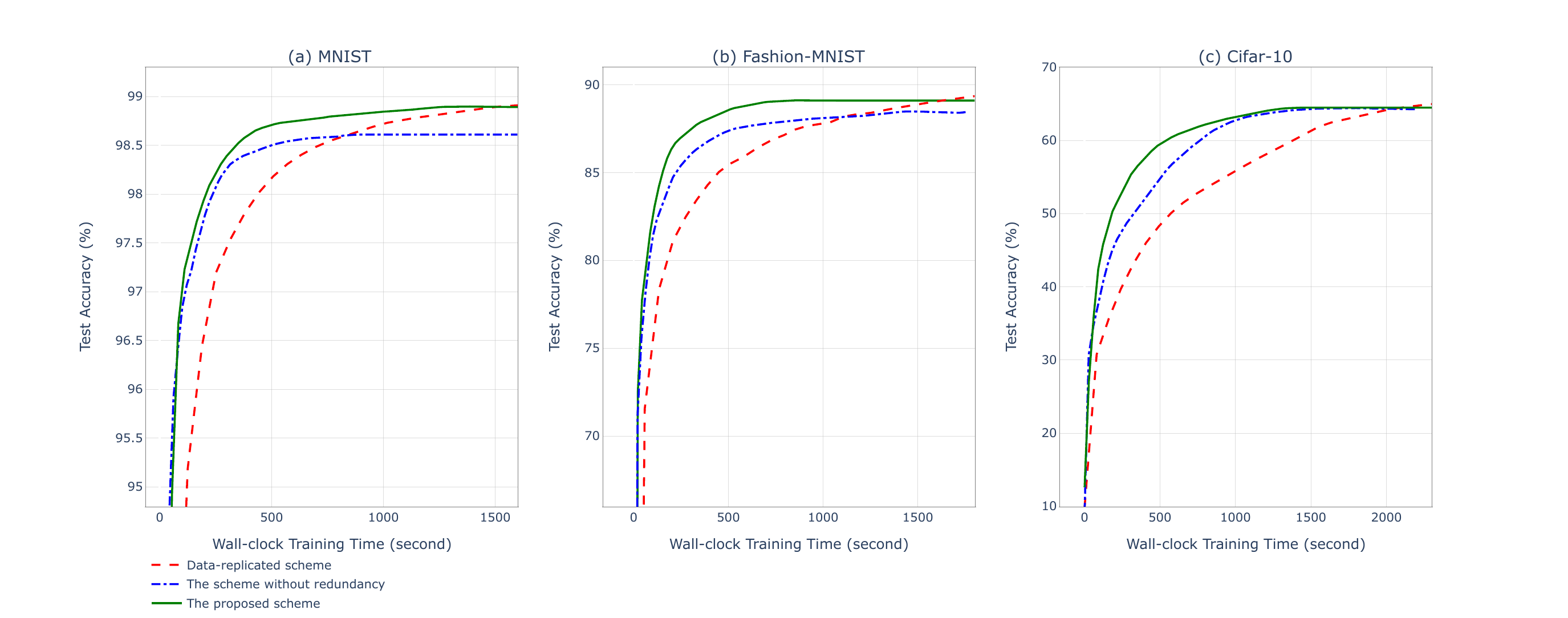}\\
		(II)\\
	\end{tabular}
	\caption{Comparison of the test accuracy of BACC, data replication scheme, and the scheme without data redundancy on (a) MNIST, (b) Fashion-MNIST, and (c) Cifar-10 on LeNet architecture, in a distributed system with (I) $N=3$, $s=1$ and (II) $N=5$, $s=2$.}
	\label{fig8}
\end{figure*}
\section{Experiments}\label{VIII}
In this section, we demonstrate the impact of BACC in distributed learning problem. The proposed scheme is evaluated for MNIST \cite{lecun2010mnist}, Fashion-MNIST \cite{xiao2017fashion}, and Cifar-10 \cite{krizhevsky2009learning} data sets. For this experiment, we use a LeNet \cite{lecun1998gradient} architecture, which consists of two convolutional layers, followed by two fully-connected layers.  In BACC, each worker node only computes gradients sampled from its coded data, which is the combination of $K$ mini-batches. More precisely,  for computing the loss function, we consider our method as a multi-label classification problem in which the training set is composed of samples each associated with a set of labels. In BACC, each worker node needs known coded labels for its training coded data set. For example, the $i$th coded label of the $r$th worker node can be considered as the normalized form of vector $\mathbf{y}_{i,r}^{\text{coded}}\triangleq\bigg|\sum_{j=0}^{K-1}\mathbf{y}_{j,r}\mu_{j,r}\bigg|,$
where $\mathbf{y}_{j,r}$ for $j\in[K-1]$ are $K$ one-hot labels of the samples which are combined with specific coefficients in the $i$th coded data sample. Thus, $\mu_{j,r}={{\frac{{(-1)}^j}{(z_r-\alpha_j)}}}/{\sum_{i=0}^{K-1}\frac{{(-1)}^i}{(z_r-\alpha_i)}}$. Note that the coded labels show that which classes are combined in the coded sample by considering their coefficients. Then, each worker node uses Sigmoid Cross-Entropy loss function to compute the loss and start the  back-propagation algorithm.  Having computed  the coded gradient vectors of the layers, each worker node sends their coded gradient vectors to the master node.  The master node then decodes all desired gradients and updates the model. 

{For comparison, we implement two other methods. One of them is a data replication-based approach, in which each mini-batch is replicated on $s+1$ worker nodes with a specific pattern such that the resulting distributed system tolerates the presence of $s$ stragglers. Another method is a distributed approach without data redundancy. Thus, it can not tolerate the presence of stragglers and can be considered as an approximated scheme.  
As a simple example, Fig. \ref{fig9} (a) shows the proposed scheme,  Fig. \ref{fig9} (b) shows the distributed scheme without data redundancy, and Fig. \ref{fig9} (c) shows the distributed setting of the data replication approach,  where $N=3$, $K=3$, and up to the $s=1$ of these worker nodes can be straggler. According to these configurations, in BACC, each worker node processes one coded data and sends the result back to the master node, but in the data replication approach, each worker node processes two raw data and completes its task by sending the results to the master node. Thus, the communication load needed in this approach for each worker node is two times greater compared to the proposed scheme. In other words, in BACC the computation and communication load per worker node is less compared to the data replication approach, but due to the proposed coding and decoding steps, the approximated results are computed in the proposed scheme. In the following, we implement the proposed scheme and two other approaches in a practical distributed learning problem.
We implement the experiments in PyTorch \cite{paszke2017automatic} using the MPI4Py \cite{dalcin2005mpi}, which is a Python package that provides the Message Passing Interface (MPI) standard for the Python programming language. All the experiments are
implemented on  a high performance computing system with Intel Xeon CPU E5-2699A v4 and up to 512GB RAM. We implement two scenarios in which 
the number of worker nodes and the stragglers are  set to $(N,s)=(3,1)$ and $(N,s)=(5,2)$. 
Figure \ref{fig8} shows how the testing accuracy varies with wall-clock run-time of training. Note that for all configurations, tests are performed on the raw test data set in the master node after each epoch. These curves show that the proposed BACC scheme achieves a certain test accuracy faster than the data replication scheme. Of course, later the data replication-based approach achieves the same accuracy and it slightly surpasses the proposed coded scheme but it is just slower.  The scheme without data redundancy converges fast but it can not achieve the final test accuracy of the data replication scheme and the proposed BACC scheme. The reason is that there are some stragglers in the distributed system, thus in the scheme without data redundancy some parts of the training data set have not been used in the training process. }

%
\bibliographystyle{ieeetr}
\bibliography{References}
\vfill
\newpage
\section{Supplementary Materials}
\subsection{Proof of Lemma \ref{lemma6}}
\renewcommand{\theequation}{S.\arabic{equation}}
\begin{proof}
	According to \eqref{berrut}, Berrut's rational interpolation has basis functions as follows 
	\begin{align}
	\ell_{i,\text{Berrut}}(x)=\frac{{\frac{{(-1)}^i}{(x-x_i)}}}{\sum_{j=0}^{n}\frac{{(-1)}^j}{(x-x_j)}}.
	\end{align}
	Thus,  the Lebesgue constant for Berrut's rational interpolant can be calculated as follows
	\begin{align}
	\Lambda_n= \max_{k\in[n-1]}\max_{x_{k+1}<x<x_k}\sum_{i=0}^{n}{|\ell_{i,\text{Berrut}}(x)|}.
	\end{align}
	Fallowing the approach in \cite{bos2011lebesgue}, we define two functions for $k\in[n-1]$ as follows
	\begin{align}
	N_k(x)\triangleq(x_k-x)(x-x_{k+1})\sum_{i=0}^{n}{\frac{1}{|x-x_i|}},
	\end{align}
	and
	\begin{align}
	D_k(x)\triangleq(x_k-x)(x-x_{k+1}){\bigg|\sum_{j=0}^{n}\frac{{(-1)}^j}{(x-x_j)}\bigg|}.
	\end{align}
	To find an upper bound for the Lebesgue constant it is enough to bound $N_k(x)$ from above and $D_k(x)$ from below. For each $x$, there exist $k$ such that $x_k>x>x_{k+1}$, so we have
	\begin{align}
	N_k(x)&=(x_k-x)(x-x_{k+1})\sum_{i=0}^{n}{\frac{1}{|x-x_i|}}\\
	&=(x_k-x_{k+1})+(x_k-x)(x-x_{k+1})\bigg( \sum_{i=0}^{k-1}{\frac{1}{x_i-x}}+\sum_{i=k+2}^{n}{\frac{1}{x-x_i}}\bigg)\\
	&=(x_k-x_{k+1})+(x-x_{k+1})\sum_{i=0}^{k-1}{\frac{x_k-x}{x_i-x}}+(x_k-x)\sum_{i=k+2}^{n}{\frac{x-x_{k+1}}{x-x_i}}\\
	&\le(x_k-x_{k+1})\bigg(1+\sum_{i=0}^{k-1}{\frac{x_k-x_{k+1}}{x_i-x_{k+1}}}+\sum_{i=k+2}^{n}{\frac{x_k-x_{k+1}}{x_k-x_i}}\bigg)\label{41}.
	\end{align}
	One can verify that \eqref{41} attains its maximum value if there exist $k
	=\bar{k}$ such that $x_k$ and $x_{k+1}$ have the maximum possible distance between in $X_n$. This happens if  $x_{\bar{k}}=\cos{\frac{{\bar{k}}\pi}{n+s}}$ and $x_{{\bar{k}}+1}=\cos{\frac{({\bar{k}}+s+1)\pi}{n+s}}$. That means all $s$ elements not included in $X_n$ are ordered consecutively.

	On the other hand, if $k$ is an even integer then we have
	\begin{align}
	D_k(x)&=(x_k-x)(x-x_{k+1}){\bigg|\sum_{j=0}^{n}\frac{{(-1)}^j}{(x-x_j)}\bigg|}\\
	\nonumber	&=(x_k-x)(x-x_{k+1})\biggl|(\frac{1}{x-x_{0}}-\frac{1}{x-x_{1}})+\dots+(\frac{1}{x-x_{k-2}}-\frac{1}{x-x_{k-1}})+(\frac{1}{x-x_{k}}-\frac{1}{x-x_{k+1}})\\\nonumber&+(\frac{1}{x_{k+3}-x}-\frac{1}{x_{k+2}-x})+\dots\biggr|.
	\end{align}
	All paired terms except $(\frac{1}{x-x_{k}}-\frac{1}{x-x_{k+1}})$ are positive for $x_{k}>x>x_{k+1}$. So, $D_k(x)$ is bounded from below as follows
	\begin{align}
	D_k(x)\ge\bigg| \bar{D}(x) + (x_k-x)(x-x_{k+1})(\frac{1}{x-x_{k}}-\frac{1}{x-x_{k+1}})\bigg|\ge\bigg| \bar{D}(x)-(x_{k}-x_{k+1})\bigg|,\label{43}
	\end{align}
	Where $\bar{D}(x)$ is a positive number for $x\in[x_{k+1},x_k]$.
	Now if $k$ is an odd integer. So, we have
	\begin{align}
	\nonumber	D_k(x)&=(x_k-x)(x-x_{k+1})\biggl|(\frac{1}{x-x_{0}}-\frac{1}{x-x_{1}})\dots+(\frac{1}{x-x_{k-1}}-\frac{1}{x-x_{k}})+(\frac{1}{x_{k+2}-x}-\frac{1}{x_{k+1}-x})+\dots\biggr|\\
	&\nonumber\ge\bigg|(x_k-x)(x-x_{k+1})\bigg((\frac{1}{x-x_{k-1}}-\frac{1}{x-x_{k}})+(\frac{1}{x_{k+2}-x}-\frac{1}{x_{k+1}-x})\bigg)\bigg|\\
	&=\bigg|(x_k-x)(x-x_{k+1})\bigg(\frac{1}{x-x_{k-1}}+\frac{1}{x_{k+2}-x}\bigg)-(x_k-x_{k+1})\bigg|=\bigg|\tilde{D}(x)-(x_k-x_{k+1})\bigg|,\label{44}
	\end{align}
	where $\tilde{D}(x)$ has positive values for $x\in[x_{k+1},x_k]$. Therefor, according to \eqref{43} and \eqref{44}, $D_k(x)$ is minimized if there exist $k
	=\bar{k}$ such that $x_k$ and $x_{k+1}$ have the maximum possible distance in $X_n$. That means   $x_{\bar{k}}=\cos{\frac{{\bar{k}}\pi}{n+s}}$ and $x_{{\bar{k}}+1}=\cos{\frac{({\bar{k}}+s+1)\pi}{n+s}}$.
	Note that the value of $n$ has no effect on the above expressions.
\end{proof}
\subsection{Proof of Theorem \ref{myWellSpaced} }\label{appC}
\begin{proof}
	According to definition \ref{wellSpaced}, it is sufficient to find some constant parameters $C,R\ge1$ such that the three conditions in Definition \ref{wellSpaced} hold for each $x_i\in X_n,i\in[n]$. \\
	\textbf{Finding $C\ge 1$:}
	We check the first condition. Let $N=n+s$, and  assume that for each $k\in[n-1]$, there exists $\alpha_k\ge k$ such that $x_k=\tilde{x}_{\alpha_k}=\cos{\frac{\alpha_k\pi}{N}}$, where $\alpha_k\in[N-1]$. Note that the nodes are ordered, i.e., $x_0<x_1<\dots<x_n$ hence $\tilde{x}_{\alpha_0}<\dots<\tilde{x}_{\alpha_k}$. So, we have
	\begin{align}\label{C1}
	\frac{x_{k+1}-x_k}{x_{k+1}-x_j}=\frac{ -\cos{\frac{\alpha_{k+1}\pi}{N}} + \cos{\frac{\alpha_k\pi}{N}}	}{-\cos{\frac{\alpha_{k+1}\pi}{N}}+\cos{\frac{\alpha_j\pi}{N}}}=
	\frac{   \sin{\frac{(\alpha_{k+1}+\alpha_k)\pi}{2N}}\sin{\frac{(\alpha_{k+1}-\alpha_k)\pi}{2N}}	}{\sin{\frac{(\alpha_{k+1}+\alpha_j)\pi}{2N}}\sin{\frac{(\alpha_{k+1}-\alpha_j)\pi}{2N}}},
	\end{align} 
	where $\alpha_j\le\alpha_k$. Now assume that there exists $1\le\beta\le s+1$ such that $\alpha_{k+1}=\alpha_k+\beta$. So, we can rewrite \eqref{C1} as follows
	\begin{align}\label{in2}
	\frac{x_{k+1}-x_k}{x_{k+1}-x_j}=\frac{   \sin{\frac{(2\alpha_k+\beta)\pi}{2N}}\sin{\frac{\beta\pi}{2N}}	}{\sin{\frac{(\alpha_{k}+\beta+\alpha_j)\pi}{2N}}\sin{\frac{(\alpha_{k}+\beta-\alpha_j)\pi}{2N}}}.
	\end{align} 
	According to the range of $k,j$ and $\beta$, we know that $\frac{\beta\pi}{2N}\le\frac{\pi}{2}$ and $\frac{(\alpha_{k}+\beta-\alpha_j)\pi}{2N}\le\frac{\pi}{2}$. Now we have two cases
	\begin{enumerate}
		\item
		if  $\frac{(2\alpha_k+\beta)\pi}{2N}\le\frac{\pi}{2}$: Since $2\alpha_k+\beta\ge\alpha_k+\beta+\alpha_j$ then $\frac{(\alpha_k+\beta+\alpha_j)\pi}{2N}\le\frac{\pi}{2}$. So, we have
		\begin{align}\label{in1}
		\frac{x_{k+1}-x_k}{x_{k+1}-x_j}\le\frac{  \frac{(2\alpha_k+\beta)\pi}{2N}\frac{\beta\pi}{2N}  }{\frac{2(\alpha_k+\beta+\alpha_j)\pi}{\pi 2N}\frac{2(\alpha_{k}+\beta-\alpha_j)\pi}{\pi 2N}}=\frac{\pi^2\beta(2\alpha_k+\beta)}{4(\alpha_k+\beta+\alpha_j)(\alpha_k+\beta-\alpha_j)}
		\end{align}
		Note that in \eqref{in1} we use Jordan's inequality $\frac{2\theta}{\pi}\le\sin{\theta}\le\theta$ for $\theta\in[0,\pi/2]$. One can verify that $\frac{(2\alpha_k+\beta)}{(\alpha_k+\beta+\alpha_j)}\le2$. Therefore,
		\begin{align}
		\frac{x_{k+1}-x_k}{x_{k+1}-x_j}\le\frac{\pi^2(s+1)}{2}\frac{1}{\alpha_k+\beta-\alpha_j}.
		\end{align}
		According to definitions, we know that $\alpha_k-\alpha_j\ge k-j$ for $k\ge j$. Hence,
		\begin{align}\label{bound1}
		\frac{x_{k+1}-x_k}{x_{k+1}-x_j}\le\frac{\pi^2(s+1)}{2}\frac{1}{k+1-j}.
		\end{align}
		\item
		if $\frac{(2\alpha_k+\beta)\pi}{2N}\ge\frac{\pi}{2}$: According to \eqref{in2} in this case we have
		\begin{align}
		\frac{x_{k+1}-x_k}{x_{k+1}-x_j}\le\frac{   \sin{\frac{(2\alpha_k+\beta)\pi}{2N}}{\frac{\beta\pi}{2N}}	}{\sin{\frac{(\alpha_{k}+\beta+\alpha_j)\pi}{2N}}{\frac{2(\alpha_{k}+\beta-\alpha_j)\pi}{\pi2N}}}.
		\end{align}
		Because $2\alpha_k+\beta\ge\alpha_k+\beta+\alpha_j$ if  $\frac{(\alpha_k+\beta+\alpha_j)\pi}{2N}\ge\frac{\pi}{2}$, then $\frac{\sin{\frac{(2\alpha_k+\beta)\pi}{2N}}}{\sin{\frac{(\alpha_k+\beta+\alpha_j)\pi}{2N}}}\le 1$. So,
		\begin{align}\label{bound2}
		\frac{x_{k+1}-x_k}{x_{k+1}-x_j}\le\frac{\pi (s+1)}{2}\frac{1}{k+1-j}.
		\end{align}
		On the other hand if $\frac{(\alpha_k+\beta+\alpha_j)\pi}{2N}\le\frac{\pi}{2}$ then by using the inequality $\sin\theta\le\frac{2\theta}{\pi}$ for $\theta\in[\pi/2,\pi]$, we have
		\begin{align}\label{bound3}
		\frac{x_{k+1}-x_k}{x_{k+1}-x_j}\le\frac{   {\frac{2(2\alpha_k+\beta)\pi}{\pi2N}}{\frac{\beta\pi}{2N}}	}{\frac{2(\alpha_{k}+\beta+\alpha_j)\pi}{\pi2N}\frac{2(\alpha_{k}+\beta-\alpha_j)\pi}{\pi2N}}\le{\pi(s+1)}\frac{1}{k+1-j}.
		\end{align} 
	\end{enumerate}
	According to \eqref{bound1}, \eqref{bound2}, and \eqref{bound3}, the first condition of Definition \ref{wellSpaced} holds with $C=\frac{\pi^2(s+1)}{2}$. Note that $s$ is independent of $N$. With the same argument one can proof the second condition of Definition \ref{wellSpaced} as well.\\ 
	\textbf{Finding $R\ge 1$:}
	We find a constant $R\ge1$ such that the third condition of Definition \ref{wellSpaced} holds, i.e., $\frac{1}{R}\le	\frac{x_{{k}+1}-x_{{k}}}{x_{{k}}-x_{{k}-1}}\le R$. According to Lemma \ref{lemma6}, in the worst case, there exists $\bar{k}$ such that $x_i=\tilde{x}_{i}=\cos{\frac{i\pi}{N}}$ for $i\in[\bar{k}]$,  and $x_{i}=\tilde{x}_{i+s}=\cos{\frac{(i+s)\pi}{N}}$ for $i\in[\bar{k}+1,n]$. Now we consider three cases as follows\\
	\textbf{Case 1:} if $k=\bar{k}$: we observe that
	\begin{align}
	\frac{x_{\bar{k}+1}-x_{\bar{k}}}{x_{\bar{k}}-x_{\bar{k}-1}}=\frac{ \cos{\frac{(\bar{k}+s+1)\pi}{N}} - \cos{\frac{\bar{k}\pi}{N}}	}{\cos{\frac{\bar{k}\pi}{N}}-\cos{\frac{{(\bar{k}-1)}\pi}{N}}}=\frac{   \sin{\frac{(2\bar{k}+s+1)\pi}{2N}}\sin{\frac{(s+1)\pi}{2N}}	}{\sin{\frac{(2\bar{k}-1)\pi}{2N}}\sin{\frac{\pi}{2N}}},
	\end{align}
	We define  $\theta=\frac{(2\bar{k}-1)\pi}{2N}$.  Furthermore, we know that   $\theta\in[\frac{\pi}{2N},\pi-\frac{(2s+5)\pi}{2N}]$. It is clear that $\frac{\pi}{2N},\frac{(s+1)\pi}{2N}\le\frac{\pi}{2}$. Thus,
	\begin{align}
	\frac{x_{\bar{k}+1}-x_{\bar{k}}}{x_{\bar{k}}-x_{\bar{k}-1}}&\le\frac{\frac{(s+1)\pi}{2N}\sin{(\theta+\frac{(s+2)\pi}{2N})}}{\frac{1}{N}\sin\theta}=\frac{(s+1)\pi}{2}\frac{\sin{\theta}\cos{\frac{(s+2)\pi}{2N}}+\cos{\theta}\sin{\frac{(s+2)\pi}{2N}}}{\sin{\theta}}\\
	&=\frac{(s+1)\pi}{2}\big(\cos{(\frac{(s+2)\pi}{2N})}+\sin{(\frac{(s+2)\pi}{2N})}\cot{\theta}\big).
	\end{align}
	According to the range of $\theta$, we know $\cot{\theta}\le\cot{\frac{\pi}{2N}}$. Therefore,
	\begin{align}
	\frac{x_{\bar{k}+1}-x_{\bar{k}}}{x_{\bar{k}}-x_{\bar{k}-1}}&\le\frac{(s+1)\pi}{2}\big( \cos{\frac{(s+2)\pi}{2N}}+\sin{\frac{(s+2)\pi}{2N}}\cot{\frac{\pi}{2N}}\big)\\
	&=\frac{(s+1)\pi}{2}\big(\frac{\sin{\frac{(s+3)\pi}{2N}}}{\sin{\frac{\pi}{2N}}}\big)\le \frac{(s+1)(s+3)\pi^2}{4}\label{58}
	\end{align}
	On the other hand,
	\begin{align}
	\frac{x_{\bar{k}+1}-x_{\bar{k}}}{x_{\bar{k}}-x_{\bar{k}-1}}&\ge\frac{2(s+1)}{\pi}\big(\cos{(\frac{(s+2)\pi}{2N})}+\sin{(\frac{(s+2)\pi}{2N})}\cot{\theta}\big)\\&\ge\frac{2(s+1)}{\pi}\big(\cos{(\frac{(s+2)\pi}{2N})}+\sin{(\frac{(s+2)\pi}{2N})}\cot{(\pi-\frac{(2s+5)\pi}{2N})}\big)\\&=\frac{2(s+1)}{\pi}\frac{\sin{\frac{(s+3)\pi}{2N}}}{\sin{\frac{(2s+5)\pi}{2N}}}\ge\frac{2(s+1)}{\pi}\frac{\sin{\frac{(s+3)\pi}{2N}}}{\sin{\frac{(2s+6)\pi}{2N}}}\ge\frac{(s+1)}{\pi},\label{61}
	\end{align}
	if $\frac{(s+3)\pi}{N}\le \pi/2$. On the other hand, if $\frac{(s+3)\pi}{N}> \pi/2$ then $\frac{\pi}{4}<\frac{(s+3)\pi}{2N}\le\frac{\pi}{2}+\frac{\pi}{2N}$. Thus 
	\begin{align}
	\frac{x_{\bar{k}+1}-x_{\bar{k}}}{x_{\bar{k}}-x_{\bar{k}-1}}\ge\frac{2(s+1)}{\pi}\frac{\sin{\frac{(s+3)\pi}{2N}}}{\sin{\frac{(2s+5)\pi}{2N}}}\ge\frac{2(s+1)}{\pi}\frac{\sin{\frac{\pi}{4}}}{1}=\frac{\sqrt{2}(s+1)}{\pi},\label{62}
	\end{align}
	if $s< N-2$.
	According to \eqref{58}, \eqref{61} and \eqref{62}, in this case, the third condition of definition \ref{wellSpaced} holds with $R=\frac{(s+1)(s+3)\pi^2}{4}$.\\
	\textbf{Case 2:} if $k=\bar{k}+s+1$, then we have
	\begin{align}\label{73}
	\frac{x_{\bar{k}+s+2}-x_{\bar{k}+s+1}}{x_{\bar{k}+s+1}-x_{\bar{k}}}=\frac{ \cos{\frac{(\bar{k}+s+2)\pi}{N}} - \cos{\frac{(\bar{k}+s+1)\pi}{N}}	}{\cos{\frac{(\bar{k}+s+1)\pi}{N}}-\cos{\frac{{\bar{k}}\pi}{N}}}=\frac{   \sin{\frac{(2\bar{k}+2s+3)\pi}{2N}}\sin{\frac{\pi}{2N}}	}{\sin{\frac{(2\bar{k}+s+1)\pi}{2N}}\sin{\frac{(s+1)\pi}{2N}}}
	\le \frac{\pi}{2(s+1)}\frac{\sin({\tilde{\theta}+\frac{(s+2)\pi}{2N}})}{\sin{\tilde{\theta}}},
	\end{align}
	where $\tilde{\theta}\triangleq\frac{(2\bar{k}+s+1)\pi}{2N}$. According to the range of $\tilde{\theta}\in[\frac{(s+3)\pi}{2N},\pi-\frac{(s+3)\pi}{2N}]$, \eqref{73} is bounded as follows
	\begin{align}
	\frac{x_{\bar{k}+s+2}-x_{\bar{k}+s+1}}{x_{\bar{k}+s+1}-x_{\bar{k}}}	\le \frac{2\pi}{2(s+1)}\cos{\frac{(s+2)\pi}{2N}}\le \frac{\pi}{s+1}.
	\end{align}
	On the other hand,
	\begin{align}
	\frac{x_{\bar{k}+s+2}-x_{\bar{k}+s+1}}{x_{\bar{k}+s+1}-x_{\bar{k}}}	\ge \frac{2}{\pi(s+1)}\frac{\sin{\frac{\pi}{2N}}}{\sin{\frac{(s+3)\pi}{2N}}}\ge \frac{4}{\pi^2(s+1)(s+3)}.
	\end{align}	
	\\ \textbf{Case 3:} if $k<\bar{k}$ or $k>\bar{k}+s+1$, one can verify that
	\begin{align}\label{70}
	\frac{2}{3\pi}\le\frac{x_{{k}+1}-x_{{k}}}{x_{{k}}-x_{{k-1}}}= \frac{ \cos{\frac{({k}+1)\pi}{N}} - \cos{\frac{{k}\pi}{N}}	}{\cos{\frac{{k}\pi}{N}}-\cos{\frac{{{(k-1)}}\pi}{N}}}\le \frac{3\pi}{2}.
	\end{align} 
	According to these cases, the third condition of definition \ref{wellSpaced} holds with $R=\frac{(s+1)(s+3)\pi^2}{4}$ for $s< N-2$. 
\end{proof}
\subsection{Proof of Theorem \ref{thBound}}\label{appB}
\begin{proof}
	From Theorem \ref{error1} we know
	\begin{align}
	\norm{r_{\text{Berrut},\mathcal{F}}(z)-g(z)}\le h(1+\lambda){\norm{g^{\prime\prime}(z)}},
	\end{align}
	if $n$ is odd, and 
	\begin{align}
	\norm{r_{\text{Berrut},\mathcal{F}}(z)-g(z)}\le h(1+\lambda)\big({\norm{g^{\prime\prime}(z)}}+\norm{g^{\prime}(z)}\big),
	\end{align}
	if $n$ is even.
	Let $\mathcal{X}=\{x_k\}_{k=0}^n$ be a set of ordered distinct interpolation points which is a subset of Chebyshev points of second kind, i.e.,  $\mathcal{X}\subset\tilde{\mathcal{X}}=\{\tilde{x}_{\alpha}\}_{\alpha=0}^N$, where $x_k=\tilde{x}_{\alpha_k}=-\cos{\frac{\alpha_k\pi}{N}}$ and $N=n+s$, $\alpha_k\ge k$.
	We define function $h(k)=x_{k+1}-x_{k}$. So,  there exist $1\le\beta\le s+1$ such that 
	$h(k)=-\cos{\frac{(\alpha_k+\beta)\pi}{N}}+\cos{\frac{\alpha_k\pi}{N}}$. One can show that $h(k)$ attains its maximum when $\frac{\alpha_k\pi}{N}=\frac{\pi}{2}-\frac{\beta\pi}{2N}$. Therefore, we have
	\begin{align}\label{h}
	h=\max_{0\le k\le n}({x_{k+1}-x_{k}})=2\sin\frac{\beta\pi}{2N}\le2\sin\frac{(s+1)\pi}{2N},
	\end{align}
	because $\sin(x)$ is increasing in $[0,\pi/2]$. On the other hands, according to Max-min inequality, the local mesh ratio is bounded as follows
	\begin{align}
	\lambda\le\min\{\max_{1\le i \le n-2}\frac{x_{i+1}-x_i}{x_i-x_{i-1}},\max_{1\le i \le n-2}\frac{x_{i+1}-x_i}{x_{i+2}-x_{i+1}} \}.
	\end{align}
	According  to Appendix \ref{appC}, we know that $\frac{x_{i+1}-x_i}{x_i-x_{i-1}}\le R$ and similarly one can prove that ${\frac{x_{i+1}-x_i}{x_{i+2}-x_{i+1}}}\le R$ as well, where $R=\frac{(s+1)(s+3)\pi^2}{4}$ and $i=[1:n-2]$. Therefore the mesh ratio is bounded and we have $\lambda\le R$. 

\end{proof}


\end{document}